\documentclass[pdflatex,sn-mathphys-num]{sn-jnl}% Math and Physical Sciences Numbered Reference Style 
\usepackage{graphicx}%
\usepackage{multirow}%
\usepackage{amsmath,amssymb,amsfonts}%
\usepackage{amsthm}%
\usepackage[numbers]{natbib}
\usepackage{mathrsfs}%
\usepackage[title]{appendix}%
\usepackage{xcolor}%
\usepackage{textcomp}%
\usepackage{manyfoot}%
\usepackage{booktabs}%
\usepackage{algorithm}%
\usepackage{algorithmicx}%
\usepackage{algpseudocode}%
\usepackage{listings}%
\usepackage{verbatim}
\usepackage{url}
\usepackage{bm}% bold math
\usepackage{bbm}
\usepackage{ulem}
\theoremstyle{thmstyleone}%
\newtheorem{thm}{\bf Theorem}%  meant for continuous numbers
\newtheorem{proposition}{\bf Proposition}% 

\theoremstyle{thmstyletwo}%
\newtheorem{example}{Example}%
\newtheorem{remark}{Remark}%

\theoremstyle{thmstylethree}%
\newtheorem{dfn}{\bf Definition}%

\DeclareMathOperator{\Tr}{Tr}
\newtheorem{lem}{\bf Lemma}
\newtheorem{cor}{\bf Corollary}

\begin{document}

\title[reversibility and causal emergence]{Dynamical Reversibility and A New Theory of Causal Emergence based on SVD}


\author*[1,2]{\fnm{Jiang} \sur{Zhang}}\email{zhangjiang@bnu.edu.cn}

\author[1,2]{\fnm{Ruyi} \sur{Tao}}

\author[2,3]{\fnm{Keng Hou} \sur{Leong}}

\author[1,2]{\fnm{Mingzhe} \sur{Yang}}

\author[2]{\fnm{Bing} \sur{Yuan}}

\affil*[1]{\orgdiv{School of Systems Science}, \orgname{Beijing Normal University}, \orgaddress{\street{No.19 Xinjiekouwai Street }, \city{Beijing}, \postcode{100876}, \country{China}}}

\affil[2]{\orgname{Swarma Research}, \orgaddress{\city{Beijing}, \postcode{100085}, \country{China}}}

\affil[3]{\orgdiv{Tsinghua-Berkeley Shenzhen Institute, Tsinghua Shenzhen International Graduate School}, \orgname{Tsinghua University}, \orgaddress{\city{Shenzhen}, \postcode{518055}, \country{China}}}


\abstract{The theory of causal emergence (CE) with effective information (EI) posits that complex systems can exhibit CE, where macro-dynamics show stronger causal effects than micro-dynamics. A key challenge of this theory is its dependence on coarse-graining method. {In this paper, we introduce a fresh concept of approximate dynamical reversibility and establish a novel framework for CE based on this. By applying singular value decomposition(SVD) to Markov dynamics, we find that the essence of CE lies in the presence of redundancy, represented by the irreversible and correlated information pathways. Therefore, CE can be quantified as the potential maximal efficiency increase for dynamical reversibility or information transmission. We also demonstrate a strong correlation between the approximate dynamical reversibility and EI, establishing an equivalence between the SVD and EI maximization frameworks for quantifying CE, supported by theoretical insights and numerical examples from Boolean networks, cellular automata, and complex networks. Importantly, our SVD-based CE framework is independent of specific coarse-graining techniques and effectively captures the fundamental characteristics of the dynamics.}}




\maketitle

\section{Introduction}
We live in a world surrounded by a multitude of complex systems. These systems engage in time-irreversible stochastic dynamics, leading to entropy production and disorder accumulation~\cite{prigogine1997end}. Despite this, there exists a belief that beneath the disorder in complex systems lie profound patterns and regularities~\cite{west2018scale,waldrop1993complexity}. Consequently, researchers endeavor to derive causal laws from these dynamic systems at a macroscopic scale, while disregarding detailed micro-level information as appropriate~\cite{landau1980statistical,Hoel2013,Hoel2017}. Ultimately, the goal is to develop an effective theory or model capable of elucidating the causality of complex systems at a macroscopic level.

Such idea can be captured by the theoretical framework known as causal emergence(CE) proposed by Hoel et al.~\cite{Hoel2013,Hoel2017,yuan2024emergence}. This framework builds upon an information-theoretic measure called Effective Information (EI)~\cite{tononi2003measuring}, which quantifies the causal influence between successive states within a Markov dynamical system. Through illustrative examples, they demonstrate that a coarse-grained Markov dynamics, when measured by EI at the macro-level, can exhibit stronger causal power than at the micro-level. Nevertheless, one of the foremost challenges is that the manifestation of CE relies on the specific manner in which we coarse-grain the system. Different coarse-graining methods may yield entirely disparate outcomes for CE \cite{yuan2024emergence}. Although this issue can be mitigated by maximizing EI~\cite{Hoel2013,Hoel2017,Zhang2023,yang2023}, the challenges such as computational complexity, the question of solution uniqueness, {ambiguity, and the non-commutativity of marginalization and abstraction operations continue to persist~\cite{eberhardt2022causal}.} Is it possible to construct a more robust theory of CE that is independent of the coarse-graining method?

Although Rosas et al.~\cite{rosas2020reconciling} proposed a new framework for CE based on integrated information decomposition theory~\cite{rosas2020reconciling,luppi2021like, williams2010nonnegative} and they utilize the synergistic information from all micro-variables across two consecutive time steps to quantify emergence, which does not require a predefined coarse-graining strategy, it still involves iterating through all variable combinations to derive synergistic information, resulting in significant computational complexity. Rosas et al. also proposed an approximate method to mitigate the complexity~\cite{rosas2020reconciling}, but it requires a pre-defined macro-variable. In addition, Barnett and Seth~\cite{barnett2023dynamical} introduced a novel framework for quantifying emergence based on the concept of dynamical independence. If the micro-dynamics are unrelated to the prediction of macro-dynamics, the complex system is considered to exhibit emergent macroscopic processes. However, this framework has only been applied to linear systems to date. Both of these methods for quantifying emergence are established on mutual information derived from data, leading to outcomes that are influenced by data distribution. Consequently, these results may not exclusively capture the ``causal'' or the dynamic essence of the system.

{Building on Hoel's theory, this paper aims to develop a CE framework grounded in an intriguing yet distinct concept: dynamical reversibility. Despite there are a bunch of discussions on the reversibility of a Markov chain and its connections with causality~\cite{faye1997causation, bernardo2023causal, bernardo2022bridging, farr2020causation, bernardo2023bridging, kathpalia2021time}, the reversibility discussed here differs from the conventional concepts. It refers to the invertibility of the transition probability matrix (TPM) in a Markov dynamics, which relates to the dynamics' ability to maintain information about the process's \textbf{states} (similar to reversible computing~\cite{Reversible_computing} or the unitary process in quantum mechanics~\cite{nielsen2010quantum}). In contrast, the conventional reversibility conception involves restoring the same \textbf{state distributions} from a stationary distribution by reversing the process. Notably, exact dynamical reversibility implies the conventional reversibility of a Markov chain, as we will demonstrate in this paper. There appears to be a contradiction between the reversibility of Markov dynamics and the conventional view that macroscopic processes are always irreversible. However, we do not assert that macro-dynamics are reversible; instead, we introduce an indicator called approximate dynamical reversibility to quantify their degree of reversibility. 
}

{
The concept of causality explored here focuses on the measure of causation in Markov dynamics rather than traditional notions involving interventions and counterfactuals. As noted by Hoel et al.~\cite{comolatti2022causal}, all measures of causation are combinations of two causal primitives: sufficiency and necessity. Sufficiency quantifies the probability of an effect $e$ occurring given that the cause $c$ occurs, while necessity measures the probability of $e$ not occurring if $c$ does not occur. A high measure of causation arises only when $c$ is both a sufficient and necessary condition for $e$. This actually implies a bijective(reversible) functional map between all possible causes ($c$ or $\neg c$) and all effects ($e$ or $\neg e$) if the causation measure is maximized.
}

{
This leads to an intriguing connection between reversibility and causality in Markov dynamics: when the TPM of a Markov chain approaches invertibility, the previous state effectively becomes both a sufficient and necessary condition  approximately for determining the state at the next time step. This point is further supported by examples found in references~\cite{Hoel2013,Hoel2017}, where EI, as a measure of causation, is maximized when the underlying dynamics are reversible. Therefore, we can re-frame the theory of CE as an endeavor to obtain a reversible macro-level dynamics by appropriately disregarding micro-level information.}

{
The close connection between EI as a measure of dynamical causality and the approximate dynamical reversibility of the same Markov chain allows us to easily arrive at a simple understanding of emergence: the emergence of causality is essentially equivalent to the emergence of reversibility. However, the concept of reversibility can provide us with even deeper insights. First, if we view a Markovian dynamics as a communication channel that transmit the information of the system's state into the future~\cite{Hoel2017}—where each state's probabilistic transition can be seen as an information pathway—then the more reversible the dynamics, the more efficient these information pathways become, meaning that the average amount of information transmitted through each pathway increases. Through singular value decomposition (SVD), we can reveal that the essence of CE lies in the presence of redundant information pathways in the system's dynamics. These pathways are linear dependent row vectors, they either transmit no information or transmit very little information (corresponding to singular vectors associated with zero or near-zero singular values). Most of the information, however, is transmitted through the less but more reversible core pathways in the dynamics (corresponding to the singular vectors associated with larger singular values). As a result, the degree of CE in dynamics can then be quantified as the potential maximal improvement in information efficiency (which is equivalent to reversibility efficiency) under certain accuracy constraints. And the optimal coarse-graining strategy for dynamics should focus on eliminating ineffective information pathways while preserving the dynamics as much as possible~\cite{Hoel2017}, in line with the maximization of EI that we will present in the paper.
}

Another related topic, the lumpability of Markov chains\cite{barreto2011lumping}, also delves into the process of coarse-graining a Markov chain\cite{zhang2019spectral} and its relationship with reversibility\cite{marin2014relations}. {This conception mainly focus on the legitimacy of a grouping method for states during the process of coarse-graining~\cite{buchholz1994exact}. A lumpable grouping method should guarantee the coarse-grained macro-dynamics being a legitimate Markov chain and its TPM, the time evolution operator should commute with the coarse-graining operator.} However, the criteria for determining the lumpability primarily focus on the consistency of the Markov dynamics rather than the causality assessed by EI, and the reversibility they {are concerned with} is not dynamical reversibility\cite{marin2014relations}. Therefore, the concepts explored in this paper serve as a {complement} to the understanding of the lumpability of Markov chains.  
\\

This paper commences by introducing an indicator designed to measure the proximity of a Markov chain to be dynamically reversible, utilizing the singular values of its TPM. Subsequently, it provides several formal definitions and mathematical theorems to establish the validity of the indicator and its close association with EI. Additionally, a novel definition for CE based on the approximate dynamical reversibility is introduced, followed by the validation of this definition and a demonstration of its equivalence to, and distinctions from, {EI maximization}-based causal emergence using multiple examples including Boolean networks, cellular automata, and complex networks. Finally, a more streamlined and potent coarse-graining method for general Markov chains, employing the singular value decomposition (SVD) of the TPM, is proposed. \\

\section{Results}
\subsection{Theories}
\subsubsection{Effective information and causal emergence}
First, we will briefly introduce Hoel et al.'s theory of Causal Emergence (CE), which is grounded in the information-theoretic measure known as Effective Information (EI). This measure was initially introduced in \cite{tononi2003measuring} and has since been employed to quantify causal emergence in the work of Hoel et al. in \cite{Hoel2013}. For a given Markov chain $\chi$ with a discrete state space $\mathcal{S}$ and the transitional probability matrix (TPM) $P$, where the element of $P$ at the $i$th row and the $j$th column, $p_{ij}$, is the conditional probability that the system transitions to state $j$ at the current time, given that it was in state $i$ at the previous time step, then $EI$ is defined as:
\begin{equation}
    \label{eqn.ei_definition}
    EI\equiv I(X_{t+1};X_t|do(X_t\sim U)),
\end{equation}
where $X_t, X_{t+1}, \forall t\geq 0$ represent the state variables defined on $\mathcal{S}$ at time step $t$ and $t+1$, respectively. The do-operator, denoted as $do(X_t\sim U)$, embodies Pearl's intervention concept as outlined in~\cite{pearl2018book}. This intervention enforces $X_t$ to adhere to a uniform (maximum entropy) distribution on $\mathcal{S}$, specifically $Pr(X_t=i)=1/N$, where $i\in \mathcal{S}$ and $N$ represents the total number of states in $\mathcal{S}$. Given that $Pr(X_{t+1}=j)=\sum_{i\in \mathcal{S}}p_{ij}\cdot Pr(X_t=i)$, the do-operator indirectly influences $Pr(X_{t+1})$ as well. %Here, $p_{ij}$ denotes the element at the intersection of the $i$th row and the $j$th column in matrix $P$, indicating the transition probability from state $i$ to state $j$. 
Consequently, the $EI$ metric quantifies the mutual information between $X_t$ and $X_{t+1}$ subsequent to this intervention, thereby measuring the strength of the causal influence exerted by $X_t$ on $X_{t+1}$. 

The rationale behind using the $do$ operator is to ensure that the $EI$ metric purely captures the characteristics of the underlying dynamics, specifically the TPM ($P$), while remaining unaffected by the actual distribution of $X_t$~\cite{yuan2024emergence}. This point can be more clear by showing another equivalent form of $EI$~\cite{Hoel2013}:
\begin{equation}
    \label{eqn.ei_definition2}
    EI=\frac{1}{N}\sum_{i=1}^N\sum_{j=1}^N p_{ij}\log\frac{p_{ij}}{\frac{\sum_{k=1}^N p_{kj}}{N}}=\frac{1}{N}\sum_{i=1}^N \left(P_i\cdot\log P_i-P_i\cdot\log \bar{P}\right)=\frac{1}{N}\sum_{i=1}^N D_{KL}(P_i||\bar{P}),
\end{equation}
where $P_i=\left(p_{i1},p_{i2},\cdots,p_{iN}\right)$ is the $i$th row vector of $P$, and $P$ can therefore be written as $P=(P_1^T,P_2^T,\cdot\cdot\cdot,P_N^T)^T$. In Equation \ref{eqn.ei_definition2}, $\cdot$ represents the scalar product between two vectors, and $\log$ is the {element-wise} logarithmic operator for vectors, $D_{KL}(\cdot||\cdot)$ is the KL-divergence between two probability distributions, and $\bar{P}=\frac{1}{N}\sum_{i}^NP_i$
is the average vector of all the $N$ row vectors of the TPM. Thus, $EI$ measures the average KL-divergence between any $P_i$ and their average $\Bar{P}$. To be noted that, this form of $EI$ is the generalized Jensen-Shannon divergence as mentioned in \cite{englesson2021generalized}. All the logarithms in Eq.\ref{eq:EI_two_terms} is base on 2. 

{The vector form representation of the TPM can indicate that the dynamics can be regarded as an information channel as pointed out in~\cite{Hoel2017}, and each row vector represents an information pathway. While, the similarities among the row vectors represent a redundancy in the dynamics, which serves as an information channel~\cite{Hoel2013}. As demonstrated in Equation \ref{eqn.ei_definition2}, $EI$ quantified the averaged differences between these row vectors.}

Furthermore, $EI$ can be decomposed into two terms~\cite{Hoel2013}:
\begin{equation}
\label{eq:EI_two_terms}
    EI=\frac{1}{N}\sum_{i=1}^N P_i\cdot\log P_i-\frac{1}{N}\sum_{i=1}^N P_i\cdot\log \bar{P}=\underbrace{-\bar{H}(P_i)}_{determinism}+\underbrace{H(\bar{P})}_{-degeneracy},
\end{equation}
the first term is determinism $-\Bar{H}(P_i)\equiv\frac{1}{N}\sum_{i=1}^N P_i\cdot\log P_i$ which measures how the current state can deterministically(sufficiently) influences the state in next time step, and the second term is non-degeneracy $H(\Bar{P})\equiv -\sum_{j=1}^N\Bar{P}_{\cdot j}\log \Bar{P}_{\cdot j}$ which measures how exactly we can infer(necessarily) the state in previous time step from the current state. In their original definition in \cite{Hoel2013}, both the determinism and the degeneracy are added $\log N$ to guarantee them to be positive. This decomposition reveals why $EI$ can measure the strength of causal effect of a Markov chain and the connections with the reversibility of dynamics because the determinism can be understood as a kind of sufficiency and the degeneracy is a kind of necessity \cite{comolatti2022causal,yuan2024emergence}.

This point can be more clear by the examples shown in Figure \ref{fig:examples}, where four cases of TPMs are shown and the $EI$ values and their normalized forms $\mathit{eff}\equiv EI/\log_2 N$ are also demonstrated below. It is not difficult to observe that as the TPM is close to an invertible matrix(a permutation matrix, see Proposition \ref{thm.P_invertible} in Supplementary \ref{sec:dynamica_reversibility}), $EI$ is larger. 

%Actually, $EI$ can reach its maximum $\log N$ when $P$ is dynamically reversible (that is $P$ is a permutation matrix according to Propositions \ref{thm:EI_maximum} and \ref{thm.P_invertible} presented in Appendix \ref{sec:theorems_EI}. In this case, $P_i$s for all $i\in [1,N]$ are orthogonal one-hot vectors, and $\bar{P}=\mathbbm{1}/N$). 

Causal emergence occurs when the coarse-grained TPM possessing larger EI than the original TPM. As shown in the example in Figure \ref{fig:examples}(d), which is the coarse-grained TPM of the example in Figure \ref{fig:examples}(c). And the degree of CE can be calculated as the difference between the EIs {as mentioned in~\cite{Hoel2013}}:
\begin{equation}
    \label{eq:dfn_CE}
    CE = EI(P')-EI(P),
\end{equation}
where, $P'$ is the coarse-grained TPM of $P$. In this example, the coarse-graining is implemented by collapsing the first three rows and columns of the TPM in Figure \ref{fig:examples}(c) into one macro-state. Thereafter, $EI(P')=1$(or $\mathit{eff}=1.0$) in (d) is clearly larger than $EI(P)=0.81$(or $\mathit{eff}=0.41$) in (c), which manifests that the strength of cause-effect in macro-level (coarse-grained TPM in (d)) is larger than the micro-level (c), thus, causal emergence occurs, and the degree of CE is $1-0.81=0.19$. 

To be noted, the extent of causal emergence could vary with changes in the coarse-graining method, and in certain cases, it might even be negative as shown in \cite{Hoel2013,yuan2024emergence}. {Thus, to quantify CE, it is essential to search for an optimal coarse-graining strategy that maximizes the EI of macro-dynamics. However, it is important to note that this optimal solution may not be unique\cite{eberhardt2022causal, e26080618}, and the best coarse-graining strategy could violate the lumpability requirement. This can result in ambiguity when merging different causal states and disrupt the commutativity between marginalization (the time evolution operator, i.e., the TPM) and abstraction (the coarse-graining operator), which is essential for maintaining consistent dynamics before and after coarse-graining, as previously discussed~\cite{eberhardt2022causal}.}
\\
\subsubsection{Dynamical reversibility}

Second, we will introduce the conception of dynamical reversibility for a Markov chain, and propose a quantitative indicator to measure the proximity for a general Markov chain to be dynamically reversible. 

% \subsubsubsection{Definitions and Properties}
\textbf{a. Definitions and Properties}

\begin{dfn}
\label{dfn:dynamical_reversibility}
    For a given markov chain $\chi$ and the corresponding TPM $P$, if $P$ simultaneously satisfies: 1. $P$ is an invertible matrix, that is, there exists a matrix $P^{-1}$, such that $P\cdot P^{-1}=I$; and 2. $P^{-1}$ is also an effective TPM of another Markov chain $\chi^{-1}$, then $\chi$ and $P$ can be called dynamically reversible.
\end{dfn}

It is important to clarify that in this context, a dynamically reversible Markov chain differs from the commonly used term ``time reversible'' Markov chain~\cite{stroock2013introduction,bernardo2023bridging}, as dynamical reversibility necessitates the ability of $P$ to be reversibly applied to each individual state, whereas the latter focuses on the reversibility of state space distributions. Actually, we can prove that the former implies the latter (Lemma \ref{thm:reversibility_implication} in Supplementary \ref{sec:dynamica_reversibility}).

Further, the following theorem states that all the Markov chains satisfying the two conditions mentioned in Definition \ref{dfn:dynamical_reversibility} are permutation matrices.
\\

\begin{proposition}
\label{thm.P_invertible}
    For a given markov chain $\chi$ and the corresponding TPM $P$, if $P$ is dynamically reversible as defined in Definition \ref{dfn:dynamical_reversibility}, if and only if $P$ is a permutation matrix.
\end{proposition}

\begin{proof}
The proof is referred to Supplementary \ref{sec:gamma}.    
\end{proof}

However, permutation matrices are rare in the whole class of all possible TPMs. That means that, in the case of a general TPM, it is not inherently reversible. Hence, an indicator that can quantify the proximity for a general TPM being reversible is required.

It seems that the rank $r$ of $P$ can be this indicator because if and only if $r<N$, the matrix $P$ is irreversible, and $P$ becomes more degenerate the smaller $r$ is. However, a non-degenerate(full rank) $P$ is not always dynamically reversible. Even if an inverse $P^{-1}$ exists, it may not function as an effective Transition Probability Matrix (TPM), as this requires all elements in $P^{-1}$ to fall between 0 and 1, alongside the fulfillment of the normalization condition (where the one-norm of the $i$th vector $(P^{-1})_i$ in $P^{-1}$ should equal one: $||(P^{-1})_i||_1=1$). While, according to Proposition \ref{thm.P_invertible}, TPMs must be permutation matrices to be dynamically reversible. Thus, the matrices being close to permutation matrices should be ``more'' reversible. One of an important observation is that all the row vectors in a permutation matrix are one-hot vectors(the vectors with only one element is 1, all other elements are zero). This characteristic can be captured by the Frobenius norm of $P$, $||P||_F$. Actually, $||P||_F$ is maximized if and only if the row vectors in $P$ are one-hot vectors (see Lemma \ref{lem.boundsize} in Supplementary \ref{sec:dynamica_reversibility}).

Therefore, the indicator that characterizes the approximate dynamical reversibility should be a kind of mixture of the rank and the Frobenius norm. While, the rank of a matrix $P$ can also be written as:

\begin{equation}
    \label{eq.rank_singular}
    r = \sum_{i=1}^N \sigma_i^0,
\end{equation}
where $\sigma_i\geq 0$ is the $i$th singular value of $P$. Furthermore, according to Lemma \ref{lem.boundsize} in Supplementary \ref{sec:dynamica_reversibility}, the Frobenius norm can be written as: 

\begin{equation}
\label{eq.frobenius_singular}
    ||P||_F^2=\sum_{i=1}^N\sigma_i^{2},
\end{equation}
which is also the sum of the squares of singular values. Both the rank and the Frobenius norm are connected through the singular values of $P$.

Formally, the approximate dynamical reversibility of $P$ can be formalized by the following definition: 
\\
\begin{dfn}
    \label{dfn.gamma}
    Suppose the transitional probability matrix (TPM) is $P$ for a markov chain $\chi$, and its singular values are $(\sigma_1\geq \sigma_2\geq\cdots\geq\sigma_N\geq 0)$, then the $\alpha${-}ordered approximate dynamical reversibility of $P$ is defined as:
    \begin{equation}
        \label{eqn.formal_def_gamma}
        \Gamma_{\alpha} \equiv \sum_i^N \sigma_i^{\alpha},
    \end{equation}
    where $\alpha\in (0,2)$ is a parameter.
\end{dfn}
% \\
Actually, $\Gamma_{\alpha}$ is the Schatten norm of $P$: $\Gamma_{\alpha}=||P||_{\alpha}^{\alpha}$ when $\alpha\geq 1$ (it is also called nuclear norm when $\alpha=1$), while it is quasinorm when $0<\alpha<1$~\cite{wiki_Schatten_norm,recht2010guaranteed,chi2019nonconvex,cui2020towards}.

This definition is reasonable to characterize the approximate dynamical reversibility because the exact dynamical reversibility can be obtained by maximizing $\Gamma_{\alpha}$ as mentioned by the following Theorem:
% \\

\begin{proposition}
    \label{thm.maximum}
    The maximum of $\Gamma_{\alpha}$ is $N$ for any $\alpha\in(0,2)$, and it can be achieved if and only if $P$ is a permutation matrix. 
\end{proposition}
\begin{proof}
    The proof can be referred by Supplementary \ref{sec:Gamma}.
\end{proof}

Further, $\Gamma_{\alpha}$ is lower bounded by $||P||_F^{\alpha}$ according to Lemma \ref{thm:lowerbound}, and this lower bound can be increased if the dynamics $P$ is more deterministic (more one-hot row vectors in $P$) according to Lemma \ref{cor:gamma_increasing}. For fully deterministic TPMs (all row vectors are one-hot vectors), $\Gamma_{\alpha}$ can be further increased as the number of orthogonal vectors become larger as claimed by Lemma \ref{thm.onehots}. In general, when $P$ is close to a permutation matrix, $\Gamma_{\alpha}$ approaches its maximum. These propositions and lemmas guarantee that $\Gamma_{\alpha}$ is a reasonable indicator to measure the approximate dynamical reversibility for any given $P$. All the mathematical proofs are given in Supplementary \ref{sec:Gamma}.

% \subsubsubsection{Determinism and degeneracy}
\textbf{b. Determinism and degeneracy}

To be noticed that by adjusting parameter $\alpha\in(0,2)$, we can make $\Gamma_{\alpha}$ more reflective of $P$'s \textbf{determinism} or \textbf{degeneracy}~\cite{Hoel2013}. When $\alpha\rightarrow 0$, $\Gamma_{\alpha}$ converges to the rank of $P$, which resembles the non-degeneracy term in the definition of $EI$ (Equation \ref{eq:EI_two_terms}) because $r$ decreases as $P$ is more and more degenerate. However, $\alpha$ is not allowed to take exact $0$ in Definition \ref{dfn.gamma} because $rank(P)$ is not a continuous function of $P$, and maximizing $rank(P)$ does not have to lead to permutation matrices. 

Similarly, $\Gamma_{\alpha}$ converges to $||P||_F^2$ when $\alpha\rightarrow 2$, but $\alpha$ does not become exactly 2 %is not allowed to take exact $2$
in Definition \ref{dfn.gamma} because the maximization of $\Gamma_{\alpha=2}$ does not imply $P$ being reversible. $||P||_F$ is comparable with the determinism term in the definition of $EI$ (Equation \ref{eq:EI_two_terms}) because when there are more and more one-hot row vectors in $P$, the maximum transitional probability in $P$ become larger and larger which means the underlying dynamics becomes more deterministic. 

In practice we always take $\alpha=1$ to balance the propensity of $\Gamma_{\alpha}$ for measuring determinism and degeneracy, and $\Gamma_{\alpha=1}$ is called nuclear norm which has many potential applications~\cite{fazel2002matrix,cui2020towards}. When $\alpha<1$ the measure $\Gamma_{\alpha}$ quantifies more on the non-degeneracy of $P$. In literatures, $\Gamma_{\alpha\leq 1}$ is always utilized as an approximation of the rank function~\cite{liu2014exact, nie2012low}. On the other side, the measure $\Gamma_{\alpha}$ characterizes more on the determinism of $P$ when $\alpha>1$.

Considering the importance of $\alpha=1$, we mostly show the results on $\alpha=1$, and we abbreviate $\Gamma_1$ as $\Gamma$ in the following texts.

% \subsubsubsection{Normalization and Examples}
\textbf{c. Normalization and Examples}

Since $\Gamma_{\alpha}$ is size-dependent, we need to normalize them by dividing the size of $P$
\begin{equation}
    \label{eqn:normalized_Gamma_alpha}
    \gamma_{\alpha}=\frac{\Gamma_{\alpha}}{N},
\end{equation}
to characterize the size-independent approximate dynamical reversibility such that the comparisons between Markov chains with different sizes are more reasonable. It can be proven that $\gamma_{\alpha}$ is always smaller than {or equal to} $1$ as a derived result of Proposition \ref{thm.maximum}. 

{This quantity also evaluates the averaged dynamical reversibility, or the efficiency of information transmission through $P$ by each information pathway, treating the Markov dynamics as an information channel as stated in~\cite{Hoel2017}, and each state's transitions as an information pathway.} 

In Figure \ref{fig:examples}, we show $\Gamma_{\alpha}$s and the normalized ones, $\gamma_{\alpha}$s, on the four examples by setting $\alpha=1$. $\Gamma$ varies from 2 (case c, d) to 3.81 (case a), and $\gamma=\Gamma/N$ varies from 0.5 (case c) to 1 (case d) in Figure \ref{fig:examples}. It is clear that $\gamma$ is larger if the TPM is closer to a reversible matrix. And the correlation between $\gamma$ and $\mathit{eff}$ can be observed in these examples.

%In literateurs, $\Gamma$ is called nuclear norm or $p$-Schatten norm with $p=1$, and it is always treated as the continuous relaxation of the rank of matrix $P$.

%The reason why $\Gamma$ can quantify the reversibility of $P$ is clear if we make an SVD decomposition on $P$:

%We actually can prove strictly that 1)$\Gamma$ achieves its maximal value $N$ when $P$ is reversible (Theorem \ref{thm.upperbound_gamma} and Theorem \ref{thm.P_invertible}, see Appendix \ref{sec:gamma}); and 2)If $P$ is closed to be reversible, then $\Gamma$ also approaches its maximal value $N$ (Theorem \ref{thm:gamma_increasing}, see Appendix \ref{sec:gamma}). 

% \\

\subsubsection{Connections between $\Gamma_{\alpha}$ and $EI$}
\label{sec:comparison_EI_gamma}

On one side, $EI$ characterizes the strength of causal effect of a Markov chain; on the other side, $\Gamma_{\alpha}$ can quantitatively capture the approximate dynamical reversibility of the Markov chain. We have claimed that the causality and reversibility are deeply connected in the introduction, thus, we will discuss the connections between $EI$ and $\Gamma_{\alpha}$ in this sub-section.

First, we found that $EI$ and $\log\Gamma_{\alpha}$ share the same minimum and maximum as mentioned in the following theorem.
\\

\begin{proposition}
    \label{thm:synthesize_theorem}
    For any TPM $P$ and $\alpha\in(0,1)$, both the logarithm of $\Gamma_{\alpha}$ and $EI$ share identical minimum value of $0$ and one common minimum point at $P=\frac{1}{N}\mathbbm{1}_{N\times N}$. They also exhibit the same maximum value of $\log N$ with maximizing points corresponding to $P$ being a permutation matrix. Where the notation $\mathbbm{1}_{N\times N}$ denotes a matrix where all elements are equal to 1.
\end{proposition}
\begin{proof}
    The proof can be referred to Supplementary \ref{sec:compare_gamma_EI_app}.
\end{proof}

Thus, $\log\Gamma_{\alpha}$ and $EI$ can reach their maximal values $\log N$ when $P$ is reversible (permutation matrix). They also achieve their minimal values ($0$) when $P_i=\mathbbm{1}/N, \forall i\in\{1,2,\cdots,N\}$, where $\mathbbm{1}\equiv(1,1,\cdots,1)$. However, we can prove that $\mathbbm{1}/N$ is not the unique minimum point of $EI$, any TPM with $P_i=P_j$ for any $i,j\in \{1,2,\cdots,N\}$ can make $EI=0$ (see Lemma \ref{thm:EI_derivitive} and Corollary \ref{cor:EI_convex} in Supplementary \ref{sec:theorems_EI}).

Second, $EI$ is upper and lower bounded by an affine term of $\log\Gamma_{\alpha}$. This point can be formally stated as the following theorems.
\\
\begin{thm}
    \label{thm.EI_boundby_gamma}
    For any TPM $P$, its effective information $EI$ is upper bounded by $\frac{2}{\alpha}\log\Gamma_{\alpha}$, and lower bounded by $\log\Gamma_{\alpha}-\log N$. 
\end{thm}
\begin{proof}
    The proof can be referred to Supplementary \ref{sec:compare_gamma_EI_app}.
\end{proof}

Therefore, we have the following inequality:
\begin{equation}
    \label{eq.EI_bounds}
    \log\Gamma_{\alpha}-\log N \leq EI\leq \frac{2}{\alpha}\log\Gamma_{\alpha}.
\end{equation}

Actually, a tighter upper bound for $EI$, $EI\leq \log\Gamma_{\alpha}$, is found empirically and numerically as the results shown in the next section. We also found that $EI$ and $\Gamma_{\alpha}$ always positively correlated in many cases. Therefore, we propose an approximate relationship exists: 

\begin{equation}
\label{eq:EI_Gamma_Approx}
    EI\sim \log\Gamma_{\alpha}.
\end{equation}

\subsubsection{A new quantification for causal emergence}
\label{sec:new_quantification}
One of the major contribution of this paper is a new quantification for causal emergence based on dynamical reversibility and singular values, and this quantification is independent on any selection of coarse-graining method. Because $\Gamma_{\alpha}$ is the summation of the $\alpha$ powers of the singular values of $P$, removing zero or approximate zero singular values does not change $\Gamma_{\alpha}$. 

First, two new definitions about causal emergence are given.
\\
\begin{dfn}
    \label{dfn:clear_emergence}
    For a given markov chain $\chi$ with TPM $P$, if $r\equiv rank(P)<N$ then \textbf{clear causal emergence} occurs in this system. And the degree of CE is 
    \begin{equation}
        \label{eq:degree_clear_emergence}
        \Delta\Gamma_{\alpha} = \Gamma_{\alpha}\cdot(1/r - 1/N).    
    \end{equation}
\end{dfn}

\begin{dfn}
    \label{dfn:vague_emergence}
    For a given markov chain $\chi$ with TPM $P$, suppose its singular values are $\sigma_1\geq \sigma_2\geq\cdots\geq\cdots\geq\sigma_N\geq 0$. For a given real value $\epsilon\in [0,\sigma_1]$, if there is an integer $i\in\{1,2,\cdots,N\}$ such that $\sigma_i>\epsilon$, then there is \textbf{vague causal emergence} with the level of vagueness $\epsilon$ occurred in the system. And the degree of CE is:
    \begin{equation}
        \label{eq:degree_vague_emergence}
        \Delta\Gamma_{\alpha}(\epsilon) = \frac{\sum_{i=1}^{r_{\epsilon}}\sigma_i^{\alpha}}{r_{\epsilon}}-\frac{\sum_{i=1}^{N}\sigma_i^{\alpha}}{N},
    \end{equation}
    where $r_{\epsilon}=\max\{i|\sigma_i>\epsilon\}$
\end{dfn}
These definitions are independent of any coarse-graining method. As a result, it represents an intrinsic and objective property of Markov dynamics. Thus, the occurrences of both clear and vague CE, as well as the extent of such emergence, can be objectively quantified.

It is not difficult to see that clear CE is a special case of vague CE when $\epsilon=0$, and it has theoretical values particularly when the singular values can be solved analytically. Furthermore, the judgement on the occurrence of CE is independent of $\alpha$ because it relates only on the rank. As a result, the notion of clear CE is only determined by $P$ and is parameter free. 

While, in practice, the threshold $\epsilon$ must be given because the singular values may approach 0 infinitely while $P$ is full rank. $\epsilon$ can be selected according to the relatively clear cut-offs in the spectrum of singular values (or the logarithm of singular values). If $\epsilon$ is very small (say $\epsilon < 10^{-10}$), we can also say CE occurs roughly. {Some indicators, e.g. effective rank, can help us to select the appropriate cut-off~\cite{thibeault2024low,antoulas2005overview}.}

As Proposition \ref{thm:CE_bounds} claims, $\Delta \Gamma_{\alpha}(\epsilon)\in [0,N-1]$ for any $\epsilon\geq 0$, and CE occurs if and only if $\Delta \Gamma_{\alpha}(\epsilon)>0$ according to Corollary \ref{eq:CE_occurence_condition}. The proposition, the corollary and the proves can be found in Supplementary \ref{sec:quantify_CE_app}. 

{Second, the rationale behind Definitions \ref{dfn:clear_emergence} and \ref{dfn:vague_emergence} for causal emergence stems from the observation that, when the coarse-graining strategy is chosen by projecting the probability masses of micro-states onto the directions aligned with the singular vectors corresponding to the largest singular values, both $\gamma_{\alpha}$ and $EI$ of the coarse-grained TPM can be increased. Consequently, setting the coarse-graining method to the vectors being parallel to the major singular vectors is a necessary condition of EI maximization, as supported by a theoretical analysis and two numeric examples demonstrated in Supplementary \ref{sec:causal_emergence_reason}.}

{
Similar to Equation \ref{eq:dfn_CE}, Equations \ref{eq:degree_clear_emergence} and \ref{eq:degree_vague_emergence} characterize the potentially maximal increase in information transfer efficiency or averaged dynamical reversibility of the TPM achieved with the least information loss through an optimal coarse-graining strategy (even if the strategy itself does not need to be explicitly determined). In this context, the threshold $\epsilon$ in Equation \ref{eq:degree_vague_emergence} can be interpreted as the precision requirement for the optimal coarse-grained macro-dynamics.
}

{In equations \ref{eq:degree_clear_emergence} and \ref{eq:degree_vague_emergence}, the efficiency is quantified using the state-averaged dynamical reversibility, $\gamma_{\alpha}(r) \equiv \sum_{i=1}^r \sigma_i / r$, where $r$ represents the effective number of states, which also corresponds to the number of effective information pathways. Alternatively, the efficiency can also be quantified by replacing $\gamma_{\alpha}(r)$ with $\log\left( \sum_{i=1}^r \sigma_i \right) / \log r$, which serves as an analog to $eff\equiv EI/\log N$ by leveraging the approximate relationship $EI \sim \log \Gamma_{\alpha}$. Under this formulation, equations \ref{eq:degree_clear_emergence} and \ref{eq:degree_vague_emergence} can be interpreted as the differences in $eff$ between macro- and micro-dynamics under the potentially optimal coarse-graining strategy, further reinforcing their role as measures of CE.}

{As a result, we can conclude that the essence of the CE phenomenon lies in the presence of redundant information pathways underlying $P$, which are non-reversible and represented by singular vectors corresponding to zero or near-zero singular values. The quantification of CE is achieved by measuring the potentially best improvement in averaged reversibility($\gamma_{\alpha}(r)$) or information transmission efficiency when these redundant channels are removed by potentially optimal coarse-graining strategy.
 }
%Secondly, both $EI$ and $\Gamma$ are correlated with the auto-correlation matrix $P\cdot P^T$. As expressed in Equation \ref{eqn.ei_definition}, $EI$ actually is the ``de-similarity'' between each $P_i,\forall i\in[1,N]$ and their average $\Bar{P}$, measured by KL-divergence. While, any entry in $P\cdot P^T$ is the similarity between $P_i$ and $P_j,\forall i,j\in[1,N]$ but measured by inner product. On the other side, all squared singular values of $P$ are the eigenvalues of $P\cdot P^T$, and $\Gamma$ is the trace of $\sqrt{P\cdot P^T}$ and it also decreases when all $P_i$s are similar(see Theorem \ref{thm.upperbound_gamma} in Appendix \ref{sec:theorems_EI}).

\subsection{Experiments}
In this section, we will show the numeric results on the comparison between $EI$ and $\Gamma_{\alpha}$, the quantification of CE under the new framework. A new method of coarse-graining based on SVD is proposed to compare the results derived from EI. In this section, we use $\Gamma$ to abbreviate $\Gamma_{\alpha=1}$ if there is no extra declarations.
\subsubsection{Comparisons of $EI$ and $\Gamma$}
% \subsubsubsection{Similarities}
\textbf{a. Similarities}

In section \ref{sec:comparison_EI_gamma}, we have derived that $EI$ is upper and lower bounded by a linear term of $\log\Gamma_{\alpha}$ and we conjecture an approximate relationship: $EI\sim \log\Gamma_{\alpha}$ in theory. We will further verify these conclusions by numeric experiments in this section.

We compare $\Gamma_{\alpha}$ and $EI$ on a variety of normalized TPMs generated by three different methods: 1) softening of permutation matrices; 2) softening of degenerate matrices; and 3) totally random matrices.
Permutation matrices consist of orthogonal one-hot row vectors, whereas the row vectors in degenerate matrices may contain repeated rows. 

The approach to soften permutation and degenerate matrices involves assigning an intensity value $p_{ij}$ to each entry positioned at the $i$th row and the $j$th column of the generated permutation matrix $P$. Here, $p_{ij} = \frac{1}{\sqrt{2\pi}\sigma} \exp{\left(-\frac{(j-c_i)^2}{\sigma^2}\right)}$, where $c_i$ represents the position of the element with a value of one in the $i$th row vector as a one-hot vector, and $\sigma$ is a parameter that regulates the softening intensity. {For detailed information, refer to} Supplementary section \ref{sec:perturbation_permutation}.

Notice that, in all cases, the final values of $p_{ij}$ are normalized by dividing by $\sum_{j=1}^N{p_{ij}}$, ensuring that $p_{ij}$ represents transitional probabilities. The results are shown in Figure \ref{fig:comparison}. The details of these generative models are in Supplementary \ref{sec:comparison_experiments}.

As shown in Figure \ref{fig:comparison}(a), (b) and (c), a positive correlation is observed on all these examples, and the approximate relationship $EI\sim \log\Gamma$ is confirmed for large $N\gg 1$. This relation is obviously observed in Figure \ref{fig:comparison}(a) and (b), but degenerates to a nearly linear relation in Figure \ref{fig:comparison}(b) since limited value region of $\Gamma$ is covered. More results on different $\alpha$ can be referred to Supplementary section \ref{sec:perturbation_permutation}.

We also show the upper and lower bounds of $EI$ by the red dashed lines in Figure \ref{fig:comparison}(a) and (b). However, in Figure \ref{fig:comparison}(c), the theoretical bounds are not visible due to the concentration of points in a small area.

Empirically, a tighter upper bound of $EI$ by $\log\Gamma_{\alpha}$ is found as the grey broken lines shown in Figure \ref{fig:comparison}. Therefore, we conjecture a new relationship, $EI\leq \log\Gamma_{\alpha}$ holds, but the rigour proof left for future works.

We also obtain an analytic solution for $EI$ and $\Gamma$ on the simplest parameterized TPM with size $N=2$ (see Supplementary \ref{sec:simplist_case}), and we show the landscape how $EI$ and $\Gamma$ are dependent on the parameters $p$ and $q$, where $p$ and $q$ are the diagonal elements for the 2*2 TPM, $P=\begin{pmatrix}p & 1-p \\1-q & q\end{pmatrix}$. It is clear that both $EI$ and $\Gamma$ attain their maximum values on the diagonal regions ($p=q=0$ or $p=q=1$, in this case, $P$ is a permutation matrix). The differences between Figure \ref{fig:comparison}(e) and (f) are apparent: 1) $\Gamma$ has a peak value when $p\approx 1-q$ but $EI$ has not because it takes same value(0) when $p\approx 1-q$ and $q\approx 1-p$; 2) a broader region with $EI\approx 0$ is observed, while the region with $\Gamma\approx 1$ is much smaller; 3) an asymptotic transition from 0 to maximal $N=2$ is observed for $\Gamma$, but not for $EI$ because $\Gamma$ is a convex function but $EI$ is not(see Corollary 1 in Supplementary Material for section \ref{sec:theorems_EI}). 

While differences exist, our observations indicate that for the majority of regions within $p$ and $q$, there is a strong correlation between $EI$ and $\Gamma_{\alpha}$ in Figure \ref{fig:comparison}(e) and (f).\\

\textbf{b. Differences}
\label{sec:EI_Gamma_Diff}

Although deep connections between $EI$ and $\Gamma_{\alpha}$ have been found, differences between these two indicators exist.  

Firstly, $EI$ quantifies the distinctions between each row vector and the average row vector of $P$ through KL-divergence as defined in Equation \ref{eqn.ei_definition2}. Put differently, $EI$ gauges the resemblances among the row vectors. Conversely, $\Gamma_{\alpha}$ assesses the dynamic reversibility, particularly as $\alpha$ approaches $0$, which correlates with the linear interdependence among the row vectors. While the linear interdependence of row vectors suggests their similarity -- meaning two identical row vectors are linearly dependent -- the reverse is not necessarily true. Consequently, $\Gamma_{\alpha}$ not only captures the similarities among row vectors but also the proximity of $P$ to a dynamically reversible matrix. In contrast, $EI$ cannot accomplish this task.

This assertion can be validated through the following numerical experiments: we can create TPMs by blending linearly dependent row vectors with linearly independent row vectors, where the number of independent vectors, or the rank, is a controlled parameter. The matrices are constructed using $r$ orthogonal one-hot vectors as the base vectors and $N-r$ real vectors. To soften the base vectors with a magnitude of $\sigma$, we employ a method similar to that in the previous section and Supplementary section \ref{sec:perturbation_permutation}. The remaining $N-r$ real vectors are generated by linearly combining the base vectors, with coefficients sampled from random real numbers uniformly distributed on the interval $[0,1]$. The size $N=50$ is fixed in these experiments.%Initially, we generate $r$ independent one-hot vectors, then soften these row vectors using the same method as described in Supplementary section \ref{sec:perturbation_permutation}, with the degree of softening determined by $\sigma$. Subsequently, we create additional row vectors by linearly combining these softened one-hot vectors with randomly chosen linear coefficients. 
We then quantify the disparity between $\Gamma$ and $EI$, with the outcomes depicted in Figure \ref{fig:comparison}(d).

It becomes evident that for small values of $r$, with increasing $\sigma${(the standard deviation of the normal distribution in the softening method mentioned in previous sub-section)}, the disparity between $\log\Gamma$ and $EI$ diminishes, given that the linear {independence} of $P$ strengthens as vectors become more distinct. This underscores that linear dependency does not equate to the similarity between row vectors. However, as the number of independent row vectors grows, if $\sigma$ remains small, $P$ converges towards a permutation matrix. Consequently, both $EI$ and $\log\Gamma$ reach identical maximum values. This elucidates why a slight bump is noticeable when $r$ is substantial.

Secondly, a significant distinction exists between $EI$ and $\Gamma_{\alpha}$ even in the scenarios where all row vectors are identical, resulting in $EI=0$ while $\Gamma_{\alpha}=||\bar{P}||^{\alpha}\cdot N^{\alpha/2}$, a quantity that can vary with $||\bar{P}||$ (refer to Lemma \ref{lem:min_rank}, Supplementary \ref{sec:gamma}). This discrepancy implies that $\Gamma_{\alpha}$ -- as opposed to $EI$ -- can provide more comprehensive insights regarding the row vectors beyond their similarity to the average row vector.

{The differences between $EI$ and $\Gamma_{\alpha}$ suggest that the linear dependency of the information pathways, represented by the row vectors in $P$, may influence both their correlations and the CE of the dynamics, but cannot be captured by EI.}

\subsubsection{Quantifying causal emergence}
Examples for clear and vague CE are shown in this section. First, to show the validity of our new framework for CE, {particularly, the equivalence to the framework of EI maximization,} several Markov dynamics on Boolean networks which have been proposed in Hoel et al.'s papers~\cite{Hoel2013,Hoel2017} are selected to test our method, and compare with the results for CE derived from the method of EI. {All the coarse-grained strategies in these examples are optimal for EI maximization as claimed by Hoel et al in reference~\cite{Hoel2013}}.

Two examples of TPMs generated from the same Boolean network model with identical node mechanism for clear emergence and vague emergence are shown in Figure \ref{fig:emergence_renormalization}(a-i), respectively. The TPM in Figure \ref{fig:emergence_renormalization}(d) is derived from the Boolean network and its node-mechanism in Figure \ref{fig:emergence_renormalization}(a) and (b) directly. Their singular value spectra are shown in Figure \ref{fig:emergence_renormalization}(e) and (h), respectively. 

There are only 4 non-zero singular values (Figure \ref{fig:emergence_renormalization}(e)) for the first example in (d), therefore, clear CE occurs, and the degree of CE is $\Delta\Gamma=0.75$. This judgement for the occurrence of CE is the same as the example mentioned in Figure 2 in ref \cite{Hoel2013}.

Vague CE can be shown on the TPM in Figure \ref{fig:emergence_renormalization}(g), which is derived from (d) by adding random Gaussian noise with strength ($std=0.03$) on the TPM in (d). As a result, the singular spectrum is obtained as shown in Figure \ref{fig:emergence_renormalization}(h). We select $\epsilon=0.2$ as the threshold such that only 4 large singular values are left. The degree of CE is $\Delta\Gamma(0.2)=0.69$. The value of $\epsilon$ is selected according to the spectrum of singular values in Figure \ref{fig:emergence_renormalization}(h) where a clear cut-off at the index of $3$ and $\epsilon=0.2$ can be observed.

Figure \ref{fig:CA}(a-f) shows another example of clear CE on a more complex Boolean network model from the reference of \cite{Hoel2013}, where 6 nodes with the same node mechanism can be grouped into 3 super-nodes to show CE. The corresponding TPM of the original Boolean network model is shown in Figure \ref{fig:CA}(c). The spectrum of the singular values is shown in Figure \ref{fig:CA}(d) where 8 non-zero values are found. The degree of this clear CE is $\Delta\Gamma=2.23$. The same judgement of the occurrence of CE is obtained as in \cite{Hoel2013}. More examples on Boolean networks in references \cite{Hoel2013} and \cite{Hoel2017} can be referred to Supplementary Section \ref{sec:example_booleans}.

The quantification for CE can be applied on complex networks (Figure \ref{fig:emergence_renormalization}(j-l)) and cellular automata (Figure \ref{fig:CA}(g-i)). The example of vague CE is shown for complex networks generated by stochastic block models(SBM) with three sets of parameters (inner or intra connection probabilities) on Figure \ref{fig:emergence_renormalization}(j-l) and the same number of blocks(communities) which is 5. The TPMs are obtained by normalizing the adjacency matrix of the network by each node's degree. One of the exampled network with 100 nodes and 5 blocks (communities) is shown in Figure \ref{fig:emergence_renormalization}(j) and its spectrum of singular values is shown in Figure \ref{fig:emergence_renormalization}(k) where a clear cut-off ($\epsilon=0.3, r_{\epsilon}=5$) can be observed at the horizontal coordinate as same as the number of blocks. We can ascertain that vague CE occurs in this network model with the degree of CE, $\Delta\Gamma(0.3)=0.56$. Two more spectra of the networks generated by the SBM with the same size and the number of blocks(5) but different parameters are shown in the same figure. 

The definition on clear CE can be applied on cellular automata to discover its local emergent structures as shown in Figure \ref{fig:CA}(g-i). In this example, we quantify clear CE for local TPMs of a cellular automaton (number 40 elementary one-dimensional cellular automaton). The local TPM is obtained by the local windows including each cell and its two neighbors. The possible spectra of singular values for these local TPMs are shown in Figure \ref{fig:CA}(h) where clear CE may or may not occur. Figure \ref{fig:CA}(i) shows the distributions of clear CEs ($\Delta\Gamma$) for all cells and time steps with the red dot markers for $\Delta\Gamma>0$. We also plot the original evolution of this automaton as the background. From this experiment, we can draw a conclusion that we can identify the local emergent structures with the quantification of causal emergence for local TPMs. \\

\subsubsection{Coarse-graining based on SVD}
Although our quantification method of CE is coarse-graining method agnostic, to compare the results against the theory based on EI, we invent a concise coarse-graining method based on SVD. The details about this method can be referred to section \ref{sec:coarse-graining}.

We test our method on all the examples shown in Figure \ref{fig:emergence_renormalization} and \ref{fig:CA}. First, for the two TPMs generated according to the same Boolean network model shown in Figure \ref{fig:emergence_renormalization}(d) and (g), their coarse TPMs are shown in Figure \ref{fig:emergence_renormalization}(f) and (i), respectively. The macro-level Boolean network model (Figure \ref{fig:emergence_renormalization}(c)) can be read out from the TPMs and the projection matrix $\Phi$. Notably, the $\Gamma$ s in the coarse TPMs are almost identical to the original ones, which means that, in this scenario, our method effectively maintains $\Gamma$ with minimal alteration. We further test the CE examples in the references \cite{Hoel2013,Hoel2017}, and almost identical coarse TPMs can be obtained.  

Second, the reduced TPM of the original one (Figure \ref{fig:CA}(a)) can be obtained by the same coarse-graining method as shown in Figure \ref{fig:CA}(e), and the projection matrix $\Phi$ is shown in (f). The coarse-grained Boolean network can be read out from the reduced TPM and the projection matrix as shown in Figure \ref{fig:CA}(b). In this example, although $\Gamma$ is much reduced (from $\Gamma=20.39$ to $\Gamma=8.0$) due to the loss of information in coarse graining, the normalized approximate dynamical reversibility increases (from $\gamma=0.32$ increases to $\gamma=1.0$). {The consistent coarse-graining results from the SVD and EI maximization methods confirm the equivalence of the two theories for CE.}

The same coarse-graining method can also be applied in complex networks generated by SBM. The reduced network with 5 nodes that is derived from the original network (Figure \ref{fig:emergence_renormalization}(j)) is shown in Figure \ref{fig:emergence_renormalization}(l). In this example a relatively large decrease of $\Gamma$ (from $\Gamma=13.30$ to $\Gamma=3.33$) and a large increase of $\gamma$ (from $\gamma=0.13$ to $\gamma=0.67$) are also observed simultaneously. This indicates a large amount of information is lost during the coarse-graining, while a relatively more effective small network model with larger normalized approximate dynamical reversibility can be obtained. 

{We also compared the results of the SVD method and the EI maximization method on SBM networks and cellular automata by calculating the similarities between the vectors representing the optimal coarse-graining strategy for EI maximization and the singular vectors associated with the largest singular values, as detailed in Supplementary Section \ref{sec:consistency_maxEI_SVD}. }

%Therefore, the occurrences of both clear and vague CE, as well as the extent of such emergence, can be objectively quantified without dependence on any specific coarse-graining method.\\
\section{Discussion}
In summary, our investigation reveals that the Schatten norm, defined as the sum of the $\alpha$ powers of all singular values of a TPM, serves as a valuable indicator of the approximate dynamical reversibility within Markov dynamics. Both theoretical insights and empirical examinations substantiate a positive association, indicating a rough correspondence between the Effective Information ($EI$) metric and the logarithm of the measure $\Gamma_{\alpha}$. Notably, it is crucial to distinguish between these two metrics. While $\Gamma_{\alpha}$ effectively quantifies the degeneracy of a TPM, particularly as $\alpha\rightarrow 0$, encompassing scenarios where all row vectors are linear correlation, $EI$ solely characterizes the similarity between these row vectors and falls short in distinguishing degenerate yet dissimilar row vectors. 

Expanding on the concept of $\Gamma_{\alpha}$, we have introduced a novel CE theory that offers a more refined definition, capturing the intrinsic properties of a system regardless of the specific coarse-graining technique employed. {By applying SVD to a TPM, we have revealed that the core of CE lies in the presence of redundant irreversible information pathways in the dynamics. Therefore, the degree of CE can be quantified by the maximum potential improvement in the averaged approximate dynamical reversibility or the efficiency of information transfer within the dynamics if the redundant information pathways are discarded.}

{
The validity of our CE framework is supported by the approximate relationship $EI \sim \log \Gamma_{\alpha}$ and the strong positive correlation observed in numerical examples. We also demonstrated the equivalence of the SVD-based and EI maximization-based frameworks for CE quantification, with the former serving as a necessary condition for the latter. All numerical experiments conducted on Boolean networks, cellular automata, and complex networks discussed in this paper support this conclusion.
}

{
Although this theoretical framework parallels the CE theory based on $EI$, quantifying CE using SVD offers greater insights. Firstly, this work demonstrated that the redundancy in Markov dynamics is evident in the correlations between different information pathways, represented by the row vectors of the TPM. While these correlations can be partially characterized by EI, the differences between $\Gamma_{\alpha}$ and $EI$ discussed in Section \ref{sec:EI_Gamma_Diff} indicate that more correlations may arise from the linear dependencies among row vectors. These cannot be quantified by EI but can be captured by $\Gamma_{\alpha}$. Secondly, due to the independency of our framework on the concrete coarse-graining strategy, the problem of finding the optimal coarse-graining strategy for EI maximization and the problem of ambiguity and the violation for the commutativity do not exist. The SVD-based quantification of CE discusses a potentially optimal increase on the efficiency of information transmission. Finally, our method implies a way to find the optimal coarse-graining strategy for maximizing $EI$: to project more probability masses of states onto the directions of singular vectors corresponding to the largest singular values, that is the basic idea of the coarse-graining method mentioned in section \ref{sec:coarse-graining}.}

{Our framework has several potential extensions. One interesting direction is to incorporate Rosas' framework of CE based on integrated information decomposition~\cite{rosas2020reconciling}. The synergetic nature of dynamics should be represented by the TPM and characterized through SVD. Another promising extension is to explore the relationship between SVD and $\Phi$, the degree of information integration, as integrated information theory is also grounded in effective information~\cite{tononi2004information,oizumi2014phenomenology}. Our work aligns with the low rank hypothesis in complex systems as discussed in ~\cite{thibeault2024low,sun2021eigen}. While these studies also utilized the SVD method to uncover emergent phenomena in complex systems, their objectives differ. This paper focuses on the transition probability matrix (TPM) that describes system dynamics, whereas Thibeault et al. analyze the coefficient matrix of complex dynamics, and Chen et al. concentrate on the matrix of fluctuations. Further research is needed to explore the relationships among these approaches. 
}

{Furthermore, this research establishes a potential link between statistical physics and artificial intelligence by framing both coarse-graining and macro-dynamics (coarse-grained TPM) as computational processes performed by an intelligent agent. Traditional statistical mechanics indicates that the information lost during coarse-graining corresponds to Boltzmann entropy~\cite{Boltzmanns_entropy_formula}. The resulting emergent macro-dynamics can be viewed as the agent's internal model representing the underlying micro-dynamics. Consequently, this work on CE suggests that the agent seeks to develop a reversible dynamical model of reality, while incurring information loss during the coarse-graining observation process. Therefore, the intelligent agent’s challenge lies in finding an optimal balance between the information loss from coarse-graining and the resulting gains in reversibility. Ironically, as quantum mechanics suggests, the fundamental world is dynamically reversible.
}

{This work has several weak points. First, the definition and quantification of emergence are abstract due to the absence of a coarse-graining strategy, which fails to capture the conventional idea that the whole is greater than the sum of its parts. Second, discussions have primarily focused on state space, while an approach incorporating variable space would be more practical, despite the exponential growth of state space size with the number of variables, posing a significant challenge. Lastly, this work relies on the TPM of dynamics, but estimating the TPM from data is difficult. }

Another weak point in this study is that the theoretical upper and lower bounds of $EI$ by $\log\Gamma_{\alpha}$ should be further studied since a tighter upper bound is empirically observed from numeric experiments, but cannot be proved. More mathematical tools and methods should be developed to study this problem. As a result, further research is warranted to address these issues.

% \backmatter

\section{Method}
% \subsection{Coarse-graining strategy based on SVD}
\label{sec:coarse-graining}
Although clear or vague CE phenomena mentioned in Sec \ref{sec:new_quantification} can be defined and quantified without coarse-graining, a simpler coarse-grained description for the original system is needed to compare with the results derived from EI. Therefore, we also provide a concise coarse-graining method based on the singular value decomposition of $P$ to obtain a macro-level reduced TPM. The basic idea is to project the row vectors $P_i,\forall i\in \{1,2,\cdots,N\}$ in $P$ onto the sub-spaces spanned by the singular vectors of $P$ such that the major information of $P$ is conserved, as well as $\Gamma_{\alpha}$ is kept unchanged. 

\subsection{Coarse-graining strategy based on SVD}
Concretely, the coarse-graining method contains five steps: 

1) We first make SVD decomposition for $P$ (suppose $P$ is irreducible and recurrent such that stationary distribution exist):
\begin{equation}
    \label{eq:svd_decomp}
    P = U\cdot\Sigma\cdot V^T, 
\end{equation}
where, $U$ and $V$ are two orthogonal and normalized matrices with dimension $N\times N$, and $\Sigma=\rm{diag}(\sigma_1,\sigma_2,\cdot\cdot\cdot,\sigma_N)$ is a diagonal matrix which contains all the ordered singular values. 

2) Selecting a threshold $\epsilon$ and the corresponding $r_{\epsilon}$ according to the singular value spectrum;

3) Reducing the dimensionality of row vectors $P_i$ in $P$ from $N$ to $r$ by calculating the following Equation:
\begin{equation}
    \label{eq:P_V_projection}
    \Tilde{P}\equiv P\cdot V_{N\times r},
\end{equation}
where $V_{N\times r}=(V_1^T,V_2^T,\cdot\cdot\cdot,V_r^T)$. 

4) Clustering all row vectors in $\Tilde{P}$ into $r$ groups by K-means algorithm to obtain a projection matrix $\Phi$, which is defined as:
\begin{equation}
    \label{eq:projection_matrix}
    \Phi_{ij}=\begin{cases} 1 &\mbox{if $\Tilde{P}_i$ is in the $j$th group}\\ 
0 & \mbox{otherwise},
\end{cases} 
\end{equation}
for $\forall i\in\{1,2,\cdots,N\}$, and $\forall j\in\{1,2,\cdots,r\}$.

5) Obtain the reduced TPM according to $P$ and $\Phi$. 

To illustrate how we can obtain the reduced TPM, we will first define a matrix called stationary flow matrix as follows:
\begin{equation}
    \label{eq:stationary_flow}
    F_{ij}\equiv \mu_i\cdot P_{ij}, \forall i,j\in\{1,2,\cdots,N\},
\end{equation}
where $\mu$ is the stationary distribution of $P$ which satisfies $P\cdot\mu=\mu$. 

Secondly, we will derive the reduced flow matrix according to $\Phi$ and $F$:
\begin{equation}
    \label{eq:F_definition}
    F'=\Phi^T\cdot F\cdot\Phi,
\end{equation}
where, $F'$ is the reduced stationary flow matrix. In fact, this matrix multiplication inherently aggregates the flux across all micro-states (nodes) within a group to derive the flux for the corresponding macro-states (nodes). Finally, the reduced TPM can be derived directly by the following formula:
\begin{equation}
    \label{eq:coarse-graining}
    P'_i=F'_i/\sum_{j=1}^{r}(F'_i)_j, \forall i\in\{1,2,\cdots,r\}.
\end{equation}
Finally, $P'$ is the coarse-grained TPM.

\subsection{Explanation}
We will explain why this coarse-graining strategy outlined in the previous sub-section works here.

In the first step, the reason why we SVD decompose the matrix $P$ is that the singular values of $P$ actually are the squared roots of the eigenvalues of $P^T\cdot P$ because: $P^T\cdot P=(P\cdot P^T)^T=(V\cdot\Sigma\cdot U^T)\cdot(U\cdot \Sigma\cdot V^T)=V\cdot\Sigma^2\cdot V^T$, thus, we will try our best to utilize the corresponding eigenvectors in $V$ to reduce $P$, because they may contain more important information about $P$.

In the third step, the eigenvectors with the largest eigenvalues are the major axes to explain the auto-correlation matrix $P^T\cdot P$. Therefore we can use PCA method to reduce the dimensionality of $P_i, \forall i\in\{1,2,\cdots,N\}$, it is equivalent to projecting $P_i$ onto the subspace spanned by the $r$ first eigenvectors. 

In the fourth step, we cluster all the row vectors $P_i,\forall i\in\{1,2,\cdots,N\}$ into $k<r$ groups according to the new feature vectors of $P_i$. Actually, according to the previous studies~\cite{ding2004k}, these $r$ major eigenvectors can be treated as the centroids of the clusters obtained by K-means algorithm for the row vectors $P_i, \forall i\in\{1,2,\cdots,N\}$. Therefore, we cluster all row vectors in $P$ by K-means algorithm, and the row vectors in one group should aggregate around the corresponding eigenvectors. 

The final step is to obtain the reduced TPM according to the clustering result or $\Phi$ in the previous step. This is a classic problem of lumping a Markov chain~\cite{marin2014relations, barreto2011lumping}. There are many lumping methods, and we adopt the one in~\cite{klein2020emergence}. {This method is employed because directly coarse-graining the TPM based on clustering results may violate lumpacity and probability normalization conditions. Instead, we coarse-grain the dynamics using stationary distributions, ensuring that both the stationary distribution and total stationary flux remain unchanged during the coarse-graining, as they function as conserved quantities like energy or material flows\cite{guo2015flow}. Consequently, we develop the method described in the previous subsection, which ensures adherence to normalization conditions and maintains the commutativity between the macro-dynamics TPM and the coarse-graining operator as proved in Proposition \ref{thm:commutativity}. }

\section*{Acknowledgements}
We wish to acknowledge all the members who participated in the ``Causal Emergence Reading Group'' organized by the Swarma Research and the ``Swarma - Kaifeng'' Workshop. And special thanks to Lixian Zhang for his insightful comments and corrections on the manuscript.

\section*{Author Contributions}
Jiang Zhang was responsible for proposing the innovative aspects of the work, designing the methodology and the experiments, and writing the main body of the paper. Ruyi Tao conducts the experiments and created all the figures for the paper. Keng Hou Leong and Yang Mingzhe  were in charge of designing and conducting some of the experiments. Bing Yuan assisted with model construction and supervised and polished the writing.

\section*{Competing Interests}
The authors of this paper: Jiang Zhang, Ruyi Tao, Keng Hou Leong, Mingzhe Yang and Bing Yuan declare no competing interests.

\section*{Data Availability Statements}
The datasets used and analyzed during the current study available from the corresponding author on reasonable request.

\section*{Code Availability Statements}
The underlying code for this study is not publicly available but may be made available to qualified researchers on reasonable request from the corresponding author.

\bibliography{sn-bibliography}

\clearpage
\section{Figures}

\begin{figure*}[h]
    \centering
    \includegraphics[width=0.8\linewidth]{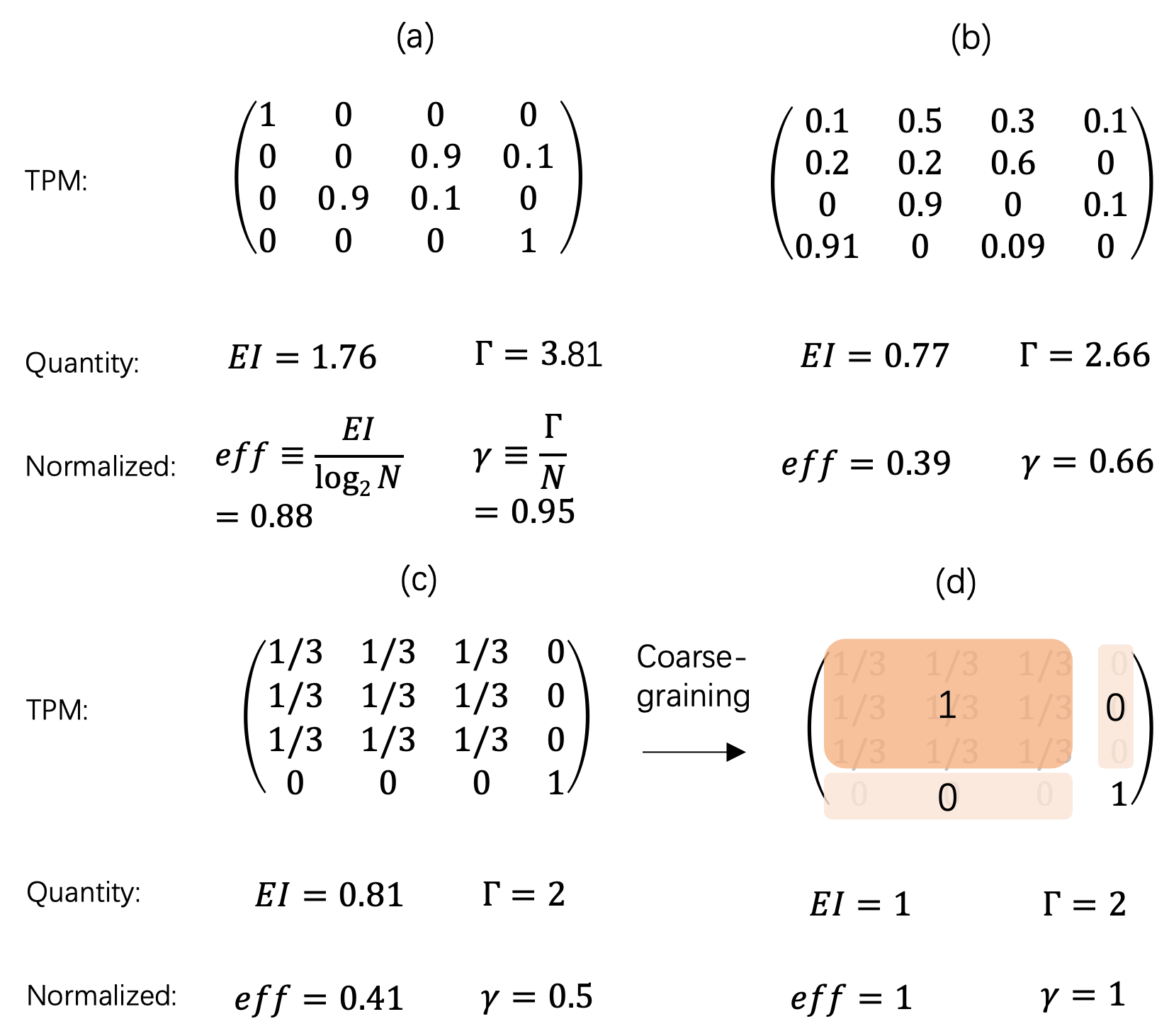}
    \caption{Examples of transitional probability matrices (TPMs) used in different Markov chains, along with measures of effective information (EI and normalized one $\mathit{eff}$) for causality, and the measure of $\Gamma_{\alpha}$ (and normalized one $\gamma_{\alpha}$ by setting $\alpha=1$) for approximate reversibility of dynamics. }
    \label{fig:examples}
\end{figure*}

\begin{figure*}[h]
    \centering
    \includegraphics[width=1\linewidth]{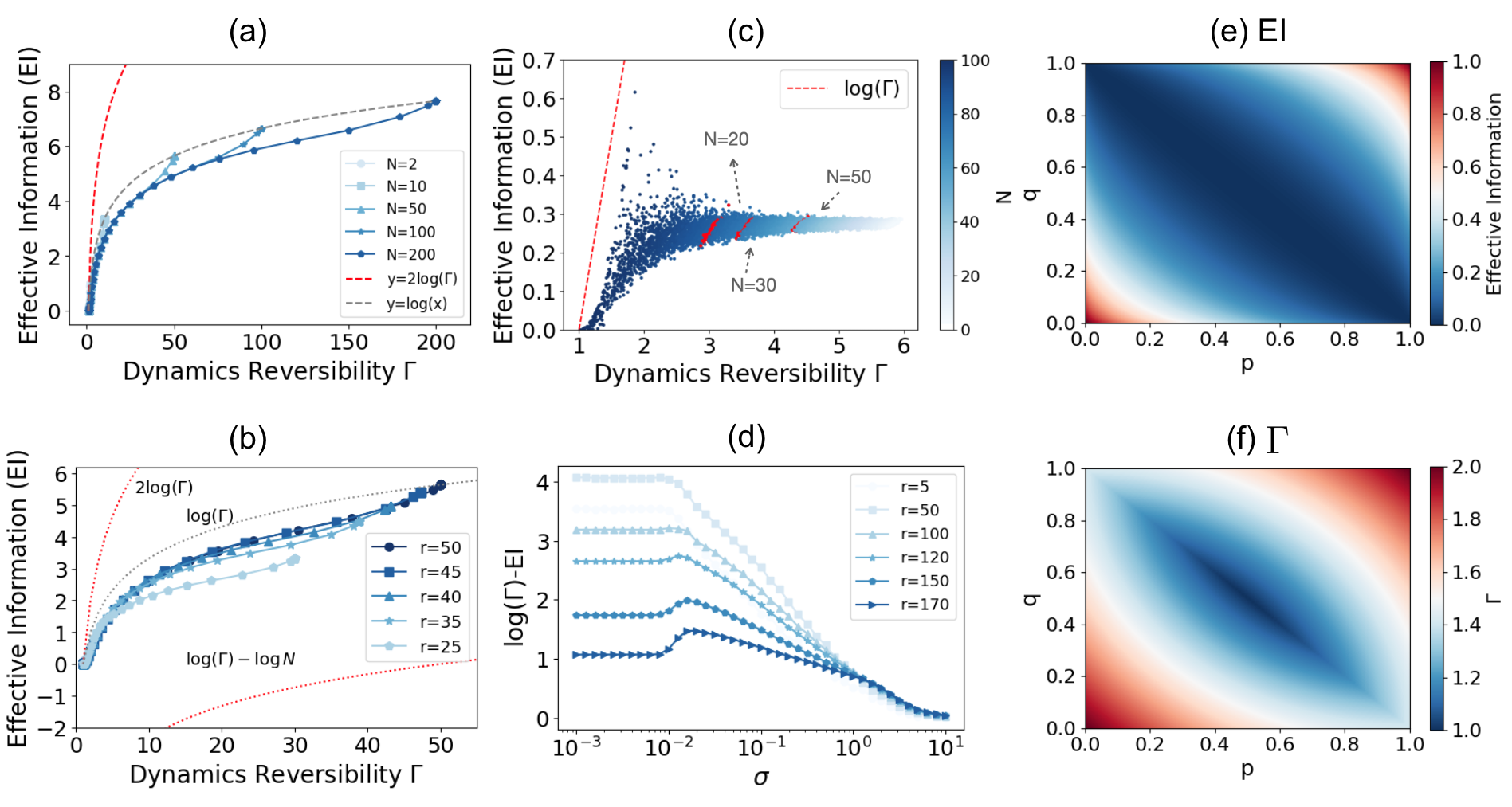}
    \caption{Comparison between $EI$ and $\Gamma$ for various generated TPMs. (a) shows the approximate relationship $EI\sim \Gamma$ for randomly generated TPMs which are softening of permutation matrices  {controlled by different $\sigma$ (detailed could be found in Section \ref{sec:perturbation_permutation})}. The theoretical and empirical upper bounds are also shown as dashed lines. The lower bound is omitted because it is size-dependent. Each curve is obtained by tuning up the magnitude of softening on a randomly generated permutation matrix with different sizes. (b) demonstrates the same relationship between $EI$ and $\Gamma$ for the TPMs which are generated by the similar softening method, but which are based on a variance of the $N\times N$ identity matrix. The variance is to change $N-r$ row vectors of identity matrix to the same one-hot vectors (with the value one as the first element), where $r$ is the rank of the matrix,  $N=50$. And the number $N-r$ can be treated as a control of the degeneracy of the TPM. On this figure, all the upper bounds and lower bounds are shown as dashed lines. (c) shows the same relationship for the combination of randomly sampled normalized row vectors for various sizes ($N\in \{2,3,\cdots,100\}$). On each size, 100 such random matrices are sampled to get the scatter points. The scatter points for particular sizes $N=20,30,50$ are rendered with red to show the nearly logarithmic relation between $EI$ and $\Gamma$. The empirical upper bound is also shown as the dashed line. (d) demonstrates the dependence of the difference between $EI$ and $\Gamma$ ($\log\Gamma-EI$) on the softening magnitude $\sigma$ for the matrices generated. (e) and (f) shows the density plots of $EI$(c) and $\Gamma$(d) with different parameters $p$ and $q$ computed for a parameterized simplest TPM: $P=\begin{pmatrix}p & 1-p \\1-q & q\end{pmatrix}$ with size $2\times 2$  }
    \label{fig:comparison}
\end{figure*}

\begin{figure}[h]
    \centering
    \includegraphics[width=1\linewidth]{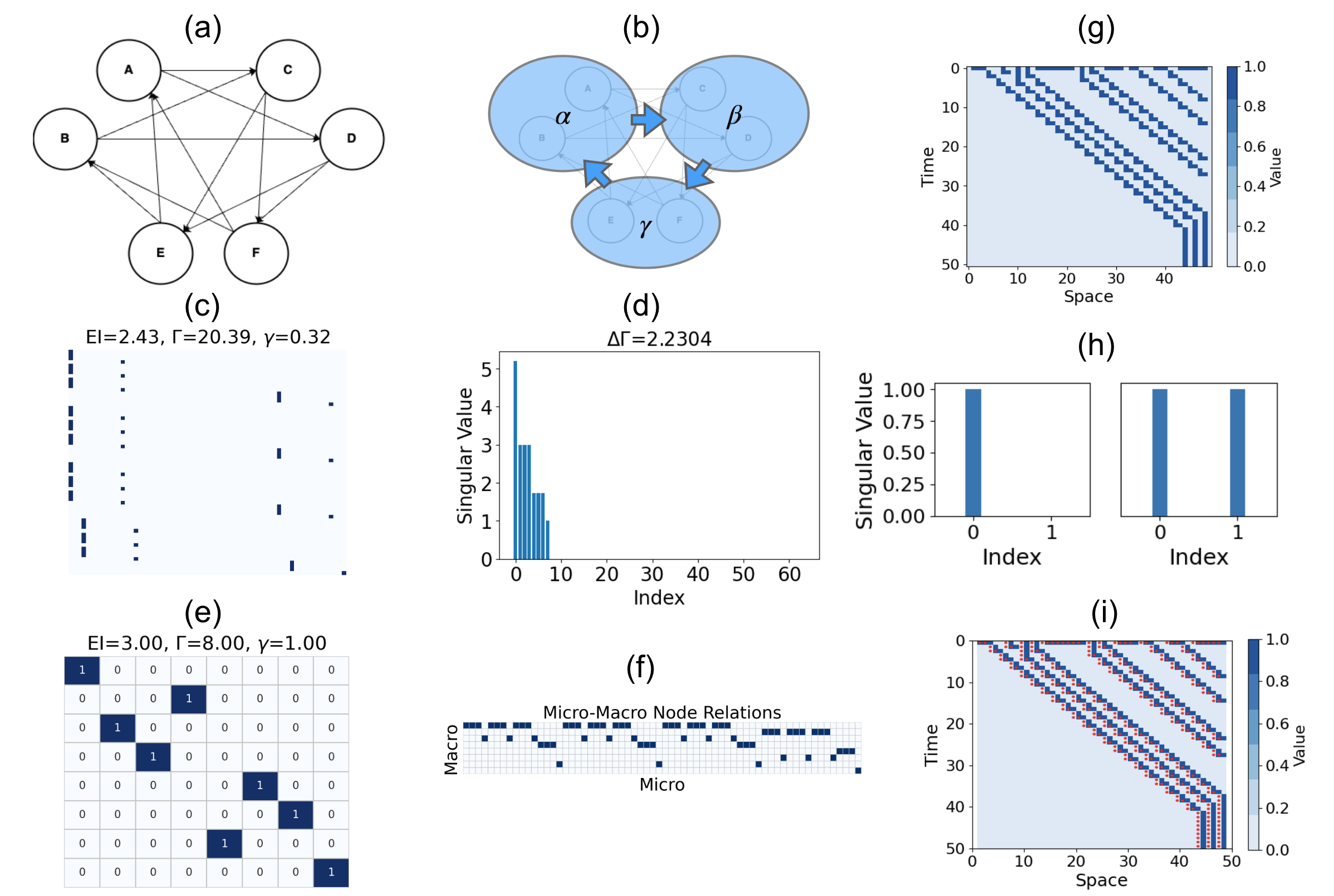}
    \caption{Examples of clear causal emergence on a Boolean network and a one-dimensional cellular automaton. (a) is the Boolean network model with 6 nodes and 12 edges, the micro-mechanism can be referred to ref.\cite{Hoel2013}. (b) is the coarse-grained Boolean network model according to the coarse-grained TPM of (e). (c) is the corresponding TPM of (a). (d) is the spectrum for the singular values of (c) with only 8 non-zero values. (e) is the coarse-graining of (c). (f) is the projection matrix from the micro-states to the macro-states obtained according to our coarse-graining method based on SVD. (g) is the evolution of the 40th elementary cellular automaton (the rule is $000\rightarrow 0, 001\rightarrow 0, 010\rightarrow 1, 011\rightarrow 0, 100\rightarrow 1, 101\rightarrow 0, 110\rightarrow 0, 111\rightarrow 0$). (h) shows the two spectra of the singular values for four distinct local TPMs of the same cellular automaton. The local TPM elucidates the process by which a cell moves from its present state to the subsequent state within a specified environment, encompassing the states of the focal cell and its two neighboring cells. (i) shows the quantification of CE for local TPM ($\Delta \Gamma$, {the red dots indicate the cells where the quantities of CE are non-zeros.}) as well as the original evolution of the original automaton (the background). }
    \label{fig:CA}
\end{figure}

\begin{figure}[h]
    \centering
\includegraphics[width=1\linewidth]{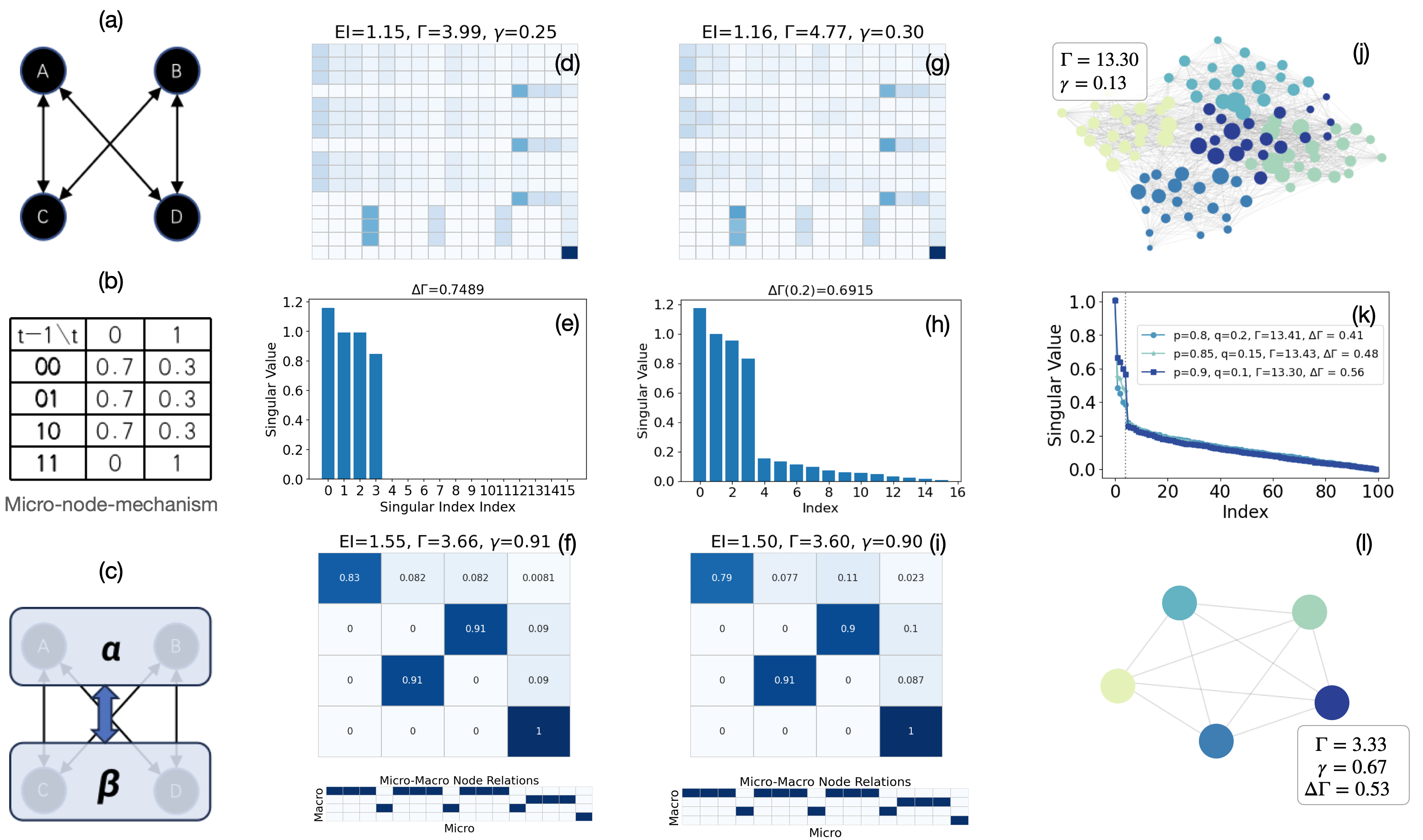}
    \caption{Examples of clear and vague CE and their coarse-grained models based on SVD method for a Boolean network and complex networks generated by the stochastic block models. (a) The original stochastic Boolean network model, each node can only interact with its network neighbors; (b) Shared node dynamics for all nodes in (a). Each row corresponds to the states combination of one node's all neighbors in previous time step, and each column is the probability to take 0 or 1 of the node at current time step. (c) The coarse-grained Boolean network of (a) which is extracted from the TPMs and the relations between micro- and macro nodes illustrated in (f) and (i) by identifying the macro-state $\alpha=0,\beta=0$ as the micro-states for $(0,0,0,0),(0,0,0,1),(0,0,1,0),(0,1,0,0),(0,1,0,1),(0,1,1,0),(1,0,0,0),(1,0,0,1),(1,0,1,0)$, $\alpha=0,\beta=1$ as the micro-states for $(0,0,1,1), (0,1,1,1),(1,0,1,1)$, $\alpha=1,\beta=0$ as the micro-states for $(1,1,0,0),(1,1,0,1),(1,1,1,0)$, and $\alpha=0,\beta=1$ as the micro-states for $(1,1,1,1)$. (d) The corresponding TPM of (a) and (b). (e) The singular value spectrum for (d). (g) A perturbed TPM from (d). (h) The singular value spectrum for (g). (f) and (i) are the reduced TPMs and the projection matrices (below) after the application of our coarse-graining method on the original TPMs in (d) and (g), respectively. (j) is the visulization of the original network sampled from the stochastic block model with p = 0.9(the probability for inner community connections) and q = 0.1(the probability for inter community connections), and the nodes are colored with different hues to distinguish the blocks to which they belong. There are 5 blocks in total. The edges are undirected and binary. The TPM is obtained by normalizing the adjacency matrix by dividing {by} each node's degree. (k) The singular value spectrum of three samples of the stochastic block model network with different p and q. (l) is the reduced network of (j) obtained through our coarse-graining method, with the node grouping results aligning with the initial block settings.}
\label{fig:emergence_renormalization}
\end{figure}

% \newpage
% \newpage
\clearpage

\begin{appendices}

\section*{\centering{Supplementary Information for \\ Dynamics Reversibility and A New Theory of Causal Emergence}}

\setcounter{section}{0}
\setcounter{figure}{0}
\setcounter{equation}{0} 
\renewcommand{\thefigure}{Supplementary Figure \arabic{figure}}
\renewcommand\appendixname{Supplementary}

% \textbf{Support Information}

\appendix
\section{\label{sec:theorems_and_proves}Theorems and Proves}
\subsection{\label{sec:theorems_EI}Propositions and Proves for $EI$}
We will layout the propositions about $EI$ and give the proves for them in this sub-section.
\begin{lem}
\label{thm:EI_derivitive}
    The $EI$ can reach its minimum value $0$ if and only if the row vectors of the probability transition matrix $P$ are identical.
\end{lem}
\begin{proof}
    The $EI$ can be expressed as:
    \begin{equation}
        \label{eq:EI_complex}
        EI=\frac{1}{N}\sum_{i=1}^ND_{KL}(P_i||\Bar{P})=\frac{1}{N}\sum_{i=1}^N\sum_{j=1}^Np_{ij}\log\frac{N\cdot p_{ij}}{\sum_{k=1}^Np_{kj}},
    \end{equation}
    where, $P_i$ is the $i$th row vector in $P$, and $\log$ is the logarithm function with base $2$. Without losing generality, suppose $EI$ is differentiable on $p_{ij}$. By taking the derivative, and using the normalization condition of probability distribution $\sum_{j=1}^Np_{ij}=1$ for $1\leq i\leq N$, we obtain that:
    \begin{equation}
        \label{eq:EI_Derivative}
        \frac{\partial EI}{\partial p_{ij}}=\log\left(\frac{p_{ij}}{p_{iN}}\right)-\log\left(\frac{\Bar{p}_{\cdot j}}{\Bar{p}_{\cdot N}}\right)
    \end{equation}
    for $1\leq i\leq N$ and $1\leq j\leq N-1$, where $\Bar{p}_{\cdot j}=\frac{1}{N}\sum_{k=1}^Np_{kj}$ for $1\leq j\leq N$. Here, the placeholder $\cdot$ represents to {the} average the probabilities of $p_{kj}$ for all row indices.  Therefore, Equation \ref{eq:EI_Derivative} is equal to 0 if and only if:
    \begin{equation}
        p_{ij}=p^*_{ij}=\Bar{p}_{\cdot j}=\frac{1}{N}\sum_{k=1}^Np_{kj}
    \end{equation}
    for any $1\leq i,j\leq N$, where $p_{ij}^*$ represents the optimal solution of Equation \ref{eq:EI_Derivative}. That is to say, all the row vectors are identical: $P_i=P_j=\Bar{P}$ for any $1\leq i,j\leq N$. And the corresponding value of $EI$ is:
    \begin{equation}
        \label{eq:min_EI}
        EI_{min}=0.
    \end{equation}
    To guarantee that $EI$ is differentiable, we require that $p_{ij}>0, p_{iN}>0, \bar{p}_{\cdot j}>0,$ and $\bar{p}_{\cdot N}>0$.
\end{proof}

Therefore, $EI$ can reach its minimum value $0$ when all the row vectors are identical. For the special matrix $P=\frac{1}{N}\cdot\mathbbm{1}$, $EI$ also equals 0, however, it is not the unique minimum point. Actually, all the matrix with identical normalized probability  {row} vector can make $EI=0$.

By further taking the second order derivative of $EI$, we can prove that $EI$ is not a convex function. 
\begin{cor}
\label{cor:EI_convex}
    The second order derivative of $EI$ with respect to the distribution $p_{st}$ with $1\leq s\leq N$ and $1\leq t\leq N-1$ is:
    \begin{equation}
        \label{eq:EI_second}
        \frac{\partial^2 EI}{\partial p_{ij}\partial p_{st}}=\frac{1}{N}\cdot\left(\frac{\delta_{i,s}\delta_{j,t}}{p_{ij}}+\frac{\delta_{i,s}}{p_{iN}}-\frac{\delta_{j,t}}{N\cdot\Bar{p}_{\cdot j}}-\frac{1}{N\cdot \Bar{p}_{\cdot N}}\right),
    \end{equation}
    and the $EI(\{p_{ij}\})$ is not a convex function.
\end{cor}
\begin{proof}
By further taking the derivative of Equation \ref{eq:EI_Derivative} with respect to the distribution $p_{st}$ with $1\leq s\leq N$ and $1\leq t\leq N-1$, we obtain Equation \ref{eq:EI_second}. 

When $i=s$, the second order derivative of $EI$ is:
\begin{equation}
\label{eq:EI_sec_pos}
\begin{aligned}
    \frac{\partial^2 EI}{\partial p_{ij}\partial p_{it}}&=\frac{\delta_{j,t}}{N}\left(\frac{1}{p_{ij}}-\frac{1}{N\cdot \Bar{p}_{\cdot j}}\right)+\frac{1}{N\cdot p_{iN}}-\frac{1}{N^2\cdot \Bar{p}_{\cdot N}}\\
    &=\delta_{j,t}\frac{\sum_{k=1}^{N-1}p_{k j}-p_{ij}}{N^2\cdot p_{ij}\cdot \Bar{p}_{\cdot j}}+\frac{\sum_{k=1}^{N-1}p_{k N}-p_{iN}}{N^2\cdot p_{iN}\cdot \Bar{p}_{\cdot N}}\\
    &=\delta_{j,t}\frac{\sum_{k\neq i}p_{kj}}{N^2\cdot p_{ij}\cdot \Bar{p}_{\cdot j}}+\frac{\sum_{k\neq i}p_{k N}}{N^2\cdot p_{iN}\cdot \Bar{p}_{\cdot N}}\geq 0,
\end{aligned}
\end{equation}
but when $i\neq s$, 
\begin{equation}
    \frac{\partial^2 EI}{\partial p_{ij}\partial p_{st}}=-\frac{\delta_{j,t}}{N^2\cdot \Bar{p}_{\cdot j}}-\frac{1}{N^2\cdot \Bar{p}_{\cdot N}}< 0
\end{equation}
holds, no matter if $j=t$ or not. Therefore, the Hessian matrix of $EI$ is not positive-definite, $EI$ is not a convex function.
\end{proof}

We will further discuss the condition and the properties of the maximum of $EI$. 
\begin{lem}
\label{thm:EI_maximum}
    The EI measure can reach its maximum $\log N$ if and only if $P$ is a permutation matrix.
\end{lem}
\begin{proof}

\begin{comment}
    If the KL-divergence is treated as a function of the distributions $P_i$ and $P_j$ for $1\leq i,j\leq N$, then $D_{KL}(P_i||P_j)$ is convex\cite{https://statproofbook.github.io/P/kl-conv.html(Chapter 2.5.5)}. Furthermore,  for any $1\leq i,j\leq N$, because the domain of $D_{KL}$ is:
    \begin{equation}
        \label{eq:boundary_set}
        Dom(D_{KL})=\{(P,Q)|P,Q\in \mathcal{R}^N and |P|_1=|Q|_1=1\},
    \end{equation}
    which is also a convex set, therefore, $D_{KL}(P_i||P_j)$ attains its maximum at some extreme point of that set according to Bauer's maximum principle\cite{Kružík, Martin (2000-11-01). "Bauer's maximum principle and hulls of sets". Calculus of Variations and Partial Differential Equations. 11 (3): 321–332.}. And the extreme points of $D_{KL}$ attains when $P_i$ and $P_j$ are one-hot vectors.

    Further, for any $i,j\in[1,N]$, $D_{KL}(P_i,P_j)=0$ iff $P_i=P_j$, thus $P_i$ must orthogonal to $P_j$. In this case, 
\end{comment}
    By noticing that
    \begin{equation}
    \begin{aligned}
        \sum_{i=1}^N\sum_{j=1}^N p_{ij}\log\Bar{p}_{\cdot j}&=\sum_{j=1}^N\left(\sum_{i=1}^Np_{ij}\right)\log\Bar{p}_{\cdot j}\\
        &=\sum_{j=1}^N\Bar{p}_{\cdot j}\log \Bar{p}_{\cdot j}=-H(\Bar{P}),
    \end{aligned}
    \end{equation}
    where $H(\Bar{P})$ is the Shannon entropy of the average distribution $\Bar{P}\equiv \sum_{i=1}P_i/N$. Thus, $EI$ can also be separated as:
    \begin{equation}
    \label{eq:twoterms_EI}
    \begin{aligned}
        EI&=\frac{1}{N}\sum_{i=1}^N\sum_{j=1}^N p_{ij}\log p_{ij} - \frac{1}{N}\sum_{i=1}^N\sum_{j=1}^N \bar{p}_{\cdot j}\log \Bar{p}_{\cdot j}\\
        &=\frac{1}{N}\sum_{i=1}^N\left(-H(P_i)\right)+H(\Bar{P}).
    \end{aligned}
    \end{equation}
    Where $H(P_i)=-\sum_{j=1}^Np_{ij}\log p_{ij}$ is the Shannon entropy of the distribution $P_i$. Furthermore, we have:
    \begin{equation}
    \label{eq:P_i}
        -H(P_i)\leq 0,
    \end{equation}
    the equality holds when $P_i$ is a one hot vector, and we also have:
    \begin{equation}
    \label{eq:max_PBar}
        H(\Bar{P})\leq \log N,
    \end{equation}
    and the equality holds when $\Bar{P}=\frac{1}{N}\cdot \mathbbm{1}$. By combining these two inequalities together, we can obtain:
    \begin{equation}
    \label{eq:EI_max}
        EI\leq 0+\log N= \log N.
    \end{equation}
    The condition that make the equality in Equation \ref{eq:P_i} and the equality in Equation \ref{eq:max_PBar} hold simantanously is that all row vectors in $P$ are one-hot vectors, and they are all different such that 
    \begin{equation}
        \frac{1}{N}\sum_iP_i=\Bar{P}=\frac{1}{N}\cdot\mathbbm{1}.
    \end{equation}
    Therefore, we reach the statement claimed by this theorem that $P$ must be a permutation matrix.
    
\end{proof}

\subsection{\label{sec:gamma}
Theorems and Proves for Dynamical Reversibility and $\Gamma$}
\subsubsection{\label{sec:dynamica_reversibility}Dynamical Reversibility and Time Reversibility}
%\begin{thm}
%\label{thm.P_invertible}
\textbf{Proposition 1}. \textit{For a given Markov chain $\chi$ and the corresponding TPM $P$, with $P$ being dynamically reversible as defined in Definition \ref{dfn:dynamical_reversibility}, if and only if $P$ is a permutation matrix.}
%\end{thm}
\begin{proof}
% \color{red}
%     Suppose $P$ is dynamically reversible, then it is invertible and its matrix inverse $Q=P^{-1}$ is still a probability transition matrix.

% For any $i$, there exists a non-zero entry $p_{ij}>0$ in the $i$th row $P_i$, because $P_i$ is a non-zero vector.

% By $PQ=I$, consider the $i$th row, for $k\neq i$,
% \[\sum_{l=1}^N p_{il}q_{lk}=0\]
% However, each entry of $P$ and $Q$ is nonnegeative. So $p_{il}q_{lk}=0$ for each $l$. Especially, for $l=j$, $p_{ij}q_{jk}=0$. Thus $q_{jk}=0$. This holds for $k\neq i$.

% Then, by $QP=I$, consider the $(j,j)$ entry,
% \[
% \sum_{k=1}^N q_{jk} p_{kj} =1
% \]
% that is, $q_{ji}p_{ij}=1$. So $q_{ji}=p_{ij}=1$, and thus $p_{ik}=0$ for $k\neq i$.

% Still by $PQ=I$, consider the $i$th column, for $k\neq i$,
% \[
% \sum_{l=1}^N p_{kl}q_{li}=0
% \]
% Each entry of $P$ and $Q$ is nonnegeative. So $p_{kl}q_{li}=0$ for each $l$. Especially, for $l=j$, $p_{kj}q_{ji}=0$. Thus $p_{kj}=0$. This holds for $k\neq i$.

% Now that for any $i$, starting from $p_{ij}>0$ in the $i$th row $P_i$, it must hold that $p_{ij}=1$ and other entries are zero in the $i$th row and in the $j$th column. We conclude that $P$ is a permutation matrix.

% \begin{comment}
    
% \end{comment}
\color{black}
If $P$ is dynamical reversible, then $P$ must be an invertible matrix and $P$ is a TPM (which means all row vectors are normalized $|P_i|_1=1$).

Because $P$ is invertible, therefore, 
\begin{equation}
    P=U\cdot diag(\lambda_1,\lambda_2,\cdot\cdot\cdot,\lambda_N)\cdot U^{-1},
\end{equation}
where, $U$ is an orthonormal matrix, $\lambda_1,\lambda_2,\cdot\cdot\cdot,\lambda_N$ are the eigenvalues of $P$, and their modulus are less than or equal to 1:
\begin{equation}
    |\lambda_i|\leq 1, \forall i\in [1,N],
\end{equation}
these inequality holds because $P$ is a TPM of a Markov chain according to \cite{seabrook2023tutorial}.

Thus, the inverse of $P$ can be expressed by:
\begin{equation}
    P^{-1}=U\cdot diag(\lambda_1^{-1},\lambda_2^{-1},\cdot\cdot\cdot,\lambda_N^{-1})\cdot U^{-1},
\end{equation}

and:
\begin{equation}
    |\lambda_i|^{-1}\geq 1, \forall i\in [1,N],
\end{equation}

However, this conflicts with the conclusion that the modulus of all the eigenvalues of a TPM must be less or equals to 1 if the inequality holds strictly. Thus, we have:
\begin{equation}
\label{eqn:lambda_1}
    |\lambda_1|=|\lambda_2|=\cdot\cdot\cdot=|\lambda_N|=1.
\end{equation}

Therefore, the eigenvalues are the complex solutions for the equation of $x^N=1$. Suppose the $i$th eigenvalue is $\lambda_i$ and the corresponding eigenvector is $v_i$, therefore:

\begin{equation}
    \label{eqn:eigenvaluevector}
    P v_i=\lambda_i v_i
\end{equation}

By taking conjugate transpose on two sides of Equation \ref{eqn:eigenvaluevector} and multiplying back to this equation, we get:

\begin{equation}
    \label{eqn:times}
    v_i^{\dagger}P^{\dagger}Pv_i=\lambda_i\lambda_i^*v_i^{\dagger}v_i=v_i^{\dagger}v_i.
\end{equation}
where $^{\dagger}$ is the conjugate transpose operator, $^*$ stands for the conjugate complex number.

This equation holds for any $i$, and consider $P$ is positive real matrix, we conclude:

\begin{equation}
    \label{eqn:unitary}
    P^T\cdot P=P\cdot P^T=I
\end{equation}

Therefore, any two row vectors of $P$ must satisfy:

\begin{equation}
    \label{eqn:orthogonal}
    P_i\cdot P_j =\delta_{ij}, \forall i,j\in\{1,2,\cdots,N\}
\end{equation}

where $\cdot$ represents the scalar product between $P_i$ and $P_j$, $\delta_{ij}=1$ only if $i=j$, otherwise $\delta_{ij}=0$. While, because $1\geq p_{ik}\geq 0, \forall i,k\in\{1,2,\cdots,N\}$, thus $P_i$s must be one-hot vectors, and $P_i$ is orthogonal to $P_j$ for any $i,j\in\{1,2,\cdots,N\}$. Therefore, $P$ is a permutation matrix.

On the other hand, if $P$ is a permutation matrix, then it is invertible because permutation matrices are full rank and eigenvalues are 1. The normalization conditions $|P_i|_1=1, \forall i\in\{1,2,\cdots,N\}$ are satisfied because all row vectors $P_i$ are one-hot vectors. Therefore, $P$ is dynamical reversible.
\end{proof}
\color{black}
Next, we will prove dynamical reversibility implies time reversibility for Markov chains.
\begin{lem}
    \label{thm:reversibility_implication}
    For a recurrent discrete Markov chain $\chi$, suppose its TPM is $P$ on the space of states $\mathcal{S}$ and its stationary distribution is $\mu$, if $P$ is dynamically reversible, then $\chi$ also satisfies time reversibility.
\end{lem}
    \begin{proof}
        If $P$ is dynamically reversible, then $P$ is a permutation matrix according to Proposition \ref{thm.P_invertible}, thus, the corresponding stationary distribution must be $\mathbbm{1}/N$, where $\mathbbm{1}=(1,1,\cdots,1)$, such that any permutation on elements in $\mathbbm{1}/N$ is the same vector. 
        Because $P$ is invertible, it satisfies 
        \begin{equation}
        \label{eq:P_invertible_PT}
            P=P^{-1}=P^T,
        \end{equation}
        then we have
        \begin{equation}
            \label{eqn:check_permutation}
            p_{ij}\cdot\frac{1}{N}=p_{ji}\cdot\frac{1}{N}.
        \end{equation}
        If we denote $\mu=\frac{\mathbbm{1}}{N}$ as a row vector, Equation \ref{eqn:check_permutation} can be written in the following form:
        \begin{equation}
            \mu_i\cdot p_{ij}=\mu_j\cdot p_{ji}, \forall i,j\in\{1,2,\cdots,N\}
        \end{equation}
        This is the condition for the time reversible Markov chain\cite{stroock2013introduction}. Therefore, $\chi$ is time reversible, and $P^T$ is the TPM of the reversible process.
        %Suppose $\chi$'s time reversed Markov chain is $\chi'$, and its TPM $Q$, {where $q_{ij}\equiv P(X_t=s_i|X_{t+1}=s_j)$, $s_i$ and $s_j\in \mathcal{S}$} are the states at time steps $t$ and $t+1$, respectively. Then $P, Q, \mu$ should satisfy the detailed balance condition:
        %\begin{equation}
        %   \label{eqn:detailed_balance}
        %    P(X_{t+1},X_t)=P(X_t,X_{t+1}),
        %\end{equation}
        %By using Bayesian formulation, 
        %\begin{equation}
        %\label{eqn:bayesian}
            %P(X_{t},X_{t+1})=P(X_{t+1}|X_t)P(X_t)=P(X_{t+1})P(X_t|X_{t+1}),
        %\end{equation}
        %and let $t\rightarrow \infty$ such that $P(X_t)=P(X_{t+1})=\mu$, Equation \ref{eqn:bayesian} can be written as:
        %\begin{equation}
        %    \label{eqn:Q_definition}
        %    {
        %    q_{ij}\cdot\mu_i=p_{ji}\cdot\mu_j
        %    }
        %\end{equation}

        %If $\chi$ is time reversible, then $Q=P$, therefore:
        %\begin{equation}
        %    \label{eqn:P_reversibility}
        %    {
        %    p_{ij}\cdot{\mu_i}=p_{ji}\cdot\mu_j.
        %    }
        %\end{equation}
        %Equation \ref{eqn:P_reversibility} is the sufficient and necessary condition for $\chi$ being time reversible. 
        %Therefore, the dynamical reversibility of $\chi$ implies its time reversibility. But its reverse is not, apparently.
       
    \end{proof}

\subsubsection{\label{sec:Gamma}Approximate Dynamical Reversibility $\Gamma_{\alpha}$}
We will present the propositions and theorems and the proofs about the measure of approximate dynamical reversibility $\Gamma_{\alpha}$ in this sub-section. Before the propositions are presented, we will prove the lemma related with the proposition that will be used in the following parts.

\begin{lem}
\label{lem.normalize}
For a TPM $P=(P_1^T,P_2^T,\cdot\cdot\cdot,P_N^T)^T$, where $P_i$ is the $i$-th row vector, then:
\begin{equation}
    P_i\cdot P_j\leq 1,  \forall i,j\in\{1,2,\cdots,N\},
\end{equation}
where, $\cdot$ represents the scalar product between the vectors $P_i$ and $P_j$. The equality holds when $P_i=P_j$ and is a one-hot vector.
\end{lem}
\begin{proof}
Because $P_i$ is a probability distribution, therefore, it should satisfy normalization condition which can be expressed as:
\begin{equation}
    |P_i|_1=\sum_{j=1}^N p_{ij}=1,
\end{equation}
where, $|\cdot|_1$ is 1-norm for vector, which is defined as the summation of the absolute values of all elements. Because $P_i$ and $P_j$ are all positive vectors, we have:
\begin{equation}
    P_i\cdot P_j=\sum_{k=1}^N p_{ik}p_{jk}\leq \sum_{k=1}^N\sum_{l=1}^Np_{ik}p_{jl}=|P_i|_1\cdot|P_j|_1=1.
\end{equation}
for any $i,j\in\{1,2,\cdots,N\}$. The equality holds when 
\begin{equation}
\label{eqn:zero_condition}
    \sum_{k\neq j}p_{ik}p_{jl}=0.
\end{equation} Due to $p_{ij}>0, \forall i,j\in\{1,2,\cdots,N\}$, Equation \ref{eqn:zero_condition} holds only if $P_i=P_j$ and $P_i$ is a one-hot vector.
\end{proof}

\begin{lem}
\label{lem.boundsize}
    For a given TPM $P$, suppose its singular values are $(\sigma_1,\sigma_2,\cdot\cdot\cdot,\sigma_N)$, we have:
    \begin{equation}
    \label{eq:sigma2_upper_bound}
        \sum_{i=1}^N\sigma_i^2=\sum_{i=1}^N P_i^2=||P||_F^2\leq N,
    \end{equation}
    and the equality holds when $P_i,\forall i\in\{1,2,\cdots,N\}$ are one-hot vectors. 
\end{lem}
\begin{proof}
Because $\sigma_i^2, \forall i\in\{1,2,\cdots,N\}$ are the eigenvalues of $P\cdot P^T$, thus, $P$ can be written:
\begin{equation}
\begin{aligned}
    P\cdot P^T=U\Sigma^2 U^T,
\end{aligned}
\end{equation}
where $U$ is a orthonormal matrix with size $N$, $\Sigma^2=diag(\sigma_1^2,\sigma_2^2,\cdot\cdot\cdot,\sigma_N^2)$. Thus,
\begin{equation}
\label{eq:psquare_N}
\begin{aligned}
    \sum_{i=1}^N\sigma_i^2&=\Tr{\Sigma^2}=\Tr{(U\Sigma^2 U^T)}=\Tr{\left(P\cdot P^T\right)}\\
    &=\sum_{i=1}^NP_i^2\leq N.
\end{aligned}
\end{equation}
The last inequality holds because of Lemma \ref{lem.normalize}. And the equality holds if and only if:

\begin{equation}
    P_i\cdot P_i=1, \forall i\in\{1,2,\cdots,N\},
\end{equation}
If $P_i \cdot P_i = 1$, then $\sum_j p_{ij}^2 = 1$, indicating that $P_i$ is on a unit ball. Additionally, $|P_i|_1 = \sum_{j=1}^N p_{ij} = 1$ for all $i$, implying that $P_i$ is on the unit hyperplane. Given that $0 \leq p_{ij} \leq 1$ for all $i,j \in \{1,2,\cdots,N\}$, the intersection of the unit sphere and the unit hyperplane occurs at the corners. Consequently, $P_i$ must be a one-hot vector, where only one element is 1 and the rest are zeros.

Further, it should be noticed that:
\begin{equation}
    \sum_{i=1}^N\sigma_i^2=\sum_{i=1}^NP_i^2=\sum_{i=1}^N\sum_{j=1}^N p_{ij}^2=||P||_F^2
\end{equation}

Therefore, Equation \ref{eq:sigma2_upper_bound} holds, and the equality in Equation \ref{eq:psquare_N} holds when all row vectors are one-hot vectors.
\end{proof}

\begin{lem}
\label{lem.equality}
    For a TPM $P$, we can write it in the following way:
    \begin{equation}
        P=(P_1^T,P_2^T,\cdot\cdot\cdot,P_N^T)^T,
    \end{equation}
    where $P_i$ is the $i$-th row vector. And suppose $P$'s singular values are $\sigma_1\geq \sigma_2\geq \cdot\cdot\cdot\geq \sigma_N$. Thus, if 
    \begin{equation}
        P_i\cdot P_i=1, \forall i\in\{1,2,\cdots,N\},
    \end{equation} 
    then the singular values of $P$ satisfy:
    \begin{equation}
        \sigma_1\geq \sigma_2\geq\cdot\cdot\cdot\geq\sigma_r\geq 1
    \end{equation}
    and
    \begin{equation}
        \sigma_{r+1}=\sigma_{r+2}=\cdot\cdot\cdot=\sigma_N=0,
    \end{equation}
    where $r$ is the rank of the matrix $P$.
\end{lem}

\begin{proof}
    \color{black}
    If $P$ is formed by one-hot vectors, and the rank of $P$ is $r$, then there are $r$ different row vectors. Suppose they are $e_{i_1},e_{i_2},\ldots,e_{i_r}$, and they repeat $n_1,\ldots,n_r$ times,respectively, where $n_1+n_2+\cdots+n_r=N$.
    
    Equivalently, there are $r$ nonzero columns and there are $n_1$ ones in the $i_1$th column, $n_2$ ones in the $i_2$th column, \dots, and $n_r$ ones in the $i_r$th column. Those ones lie in different rows. So $P^T \cdot P$ is a diagonal matrix with $n_1,n_2,\ldots,n_r$ and zeros as its diagonal elements. So the nonzero singular values of $P$ are $\sigma_1=\sqrt{n_1},\sigma_2=\sqrt{n_r},\ldots,\sigma_r=\sqrt{n_r}$. Since $n_i$ is positive integer, and we can assume that $n_1 \geq n_2 \geq \cdot \cdot \cdot \geq n_r$, so

$$
 \sigma_1\geq \sigma_2\geq\cdot\cdot\cdot\geq\sigma_r\geq 1
$$
 and 

 $$ 
 \sigma_{r+1}=\sigma_{r+2}=\cdot\cdot\cdot=\sigma_N=0,
 $$
\end{proof}
\color{black}
We want to prove that it is reasonable that the proposed measure $\Gamma_{\alpha}$ to characterize the dynamical reversibility. First, we will prove that $\Gamma_{\alpha}$ is upper bounded by the system size $N$, and it can reach the maximum value if and only if $P$ is reversible.

\begin{lem}
\label{lem.upperbound_gamma}
    For a given TPM $P$, the measure of dynamical reversibility $\Gamma_{\alpha}$ for any $\alpha\in(0,2)$,is less than or equal to the size of the system $N$.  
    
\end{lem}
\begin{proof}
%Suppose any entry of the matrix $\sqrt{P\cdot P^T}$ is $q_{ij}$. It is easy to prove that $q_{ij}\geq 0$ for any $1\leq i,j\leq N$ since all the eigenvalues of $\sqrt{P\cdot P^T}$ (the singular values of $P$) are non-negative. We will prove that $q_{ij}\leq 1$. 

%According to the definition of $\sqrt{P\cdot P^T}$, the relations between the entries of the matrix $P\cdot P^T$ and $\sqrt{P\cdot P^T}$ are:
%\begin{equation}
%\begin{aligned}   
%    \label{eqn.matrix}
%    \left[\sum_{k=1}^N q_{ik}\cdot q_{kj}\right]_{N\times N}=\left(\sqrt{P\cdot P^T}\right)^2=P\cdot P^T=\left[P_i\cdot P_j\right]_{N\times N}.
%\end{aligned}
%\end{equation}
%where $P_i$ is the $i$th row vector of $P$, thus, according to Lemma \ref{lem.normalize}, we have:
%\begin{equation}
%\begin{aligned}   
%    \sum_{k=1}^N q_{ik}\cdot q_{kj}=P_i\cdot P_j\leq 1.
%\end{aligned}
%\end{equation}
%The inequality holds because of Lemma \ref{lem.normalize}. Therefore, for the diagonal elements:
%\begin{equation}
%\label{eqn.inner}
%(q_{ii})^2\leq\sum_{k=1}^Nq_{ik}q_{ki}=P_i\cdot P_i\leq 1,
%\end{equation}
%thus, $q_{ii}\leq 1$. Finally,
%\begin{equation}
%\begin{aligned}   
%    \label{eqn.diagonal}
%\Gamma&=\sum_{i=1}^N\sigma_i=\Tr{\left(\sqrt{P\cdot P^T}\right)}=\sum_{i=1}^N q_{ii}\leq N.
%\end{aligned}
%\end{equation}
Because $0< \alpha<2$, $f(x)=x^{\alpha/2}$ is a concave function, and according to Lemma \ref{lem.boundsize}, we have:
\begin{equation}
    \label{eq:upperbound}
    \begin{aligned}
        \Gamma_{\alpha}&=\sum_{i=1}^N\sigma_i^{\alpha}=N\cdot\frac{\sum_{i=1}^N\sigma_i^{2\cdot\frac{\alpha}{2}}}{N}\leq N\left(\frac{\sum_{i=1}^N\sigma_i^2}{N}\right)^{\alpha/2}\\
        &=N^{1-\alpha/2}\left(\sum_{i=1}^NP_i^2\right)^{\alpha/2}\leq N^{1-\frac{\alpha}{2}+\frac{\alpha}{2}}=N
    \end{aligned}
\end{equation}
\end{proof}
Next, we will discuss about the upper bound of $\Gamma_{\alpha}$.\\
\\
\textbf{Proposition 2}. \textit{The maximum of $\Gamma_{\alpha}$ is $N$ for any $\alpha\in(0,2)$, and it can be achieved if and only if $P$ is a permutation matrix. }
\begin{proof}
First, we will prove that when the maximum $N$ is achieved for $\Gamma_{\alpha}$, $P$ must be a permutation matrix.

Notice that the condition for $P_i$ to satisfy the equality in the second inequality of Equation \ref{eq:upperbound} is $P_i^2=1$ for all $i\in\{1,2,\cdots,N\}$.

Therefore, according to Lemma \ref{lem.equality}, there are two cases need to be discussed separately.

\textbf{Case 1}: If there are two rows of $P$ are identical: $P_i\cdot P_j=1$ for some $i,j$, then the rank of $P$ is $r<N$. Then, according to Lemma \ref{lem.equality} the first $r$ singular values satisfy:
\begin{equation}
    \sum_{i=1}^r\sigma_i^2=N,
\end{equation}
and $\sigma_i\geq 1$ for any $i\in\{1,2,\cdots,r\}$. 

Suppose there are $s\leq r$ singular values strictly larger than 1, then $\sigma_i^2>\sigma_i^{\alpha}$ for those $1\leq i\leq s$, thus:
\begin{equation}
    \sum_{i=1}^s\sigma_i^{\alpha}<\sum_{i=1}^s\sigma_i^2,
\end{equation}

And for $i>s$, the singular values are either 0 or 1, and therefore $\sigma_i=\sigma_i^2$. Finally, we have:
\begin{equation}
    \Gamma_{\alpha}=\sum_{i=1}^N\sigma_i^{\alpha}<\sum_{i=1}^N\sigma_i^2=N,
\end{equation}

Thus, in this case, the equality in Equation \ref{eq:upperbound} does not hold.

\textbf{Case 2}: If $P_i\cdot P_j\neq 1$ for any $i\neq j$, and $P_i\cdot P_i=1$ for $\forall i\in\{1,2,\cdots,N\}$, then according to Lemma \ref{lem.equality}, if and only if all the singular values are $1$, that is, 
\begin{equation}
\label{eq:optimization_condition}
    \sigma_i=\sigma_j=1, \forall i,j\in\{1,2,\cdots,N\},
\end{equation}
which implies $P$ is invertible. Notice that, this is exactly the same condition to make the equality holds for the first inequality in Equation \ref{eq:upperbound}.

Second, we will prove the necessity for the maximum of $\Gamma_{\alpha}$. If $P$ is a permutation matrix, the eigenvalues of $P$ are the singular values because $P$ is symmetric. And the singular values satisfy:
\begin{equation}
    \label{eq:permutaiton_sigma}
    \sigma_1=\sigma_2=\cdot\cdot\cdot=\sigma_N=1,
\end{equation}
thus,
\begin{equation}
\Gamma_{\alpha}=\sum_{i=1}^N\sigma_i^{\alpha}=N
\end{equation}
which achieves the maximum value of $\Gamma_{\alpha}$.

%And the condition for $q_{ij}$ that makes the equality holds is:
%\begin{equation}
%\begin{aligned}   
%    \label{eqn.condition}
%    q_{ii}=1\land q_{ik}\cdot q_{kj}=0
%\end{aligned}
%\end{equation}
\end{proof}
\begin{comment}
\begin{cor}
    If $P_i$ in $P=(P_1,P_2,\cdot\cdot\cdot,P_N)^T$ is one hot vector for $\forall i\in[1,N]$, and $P$ is degenerative, then the following equation holds:
    \begin{equation}
        \Gamma_{\alpha}<\sum_{i=1}^N\sigma_i^2=N
    \end{equation}
    for any $\alpha\in(0,2)$
\end{cor}
\begin{proof}
    Because $P_i$ are all one hot vectors, thus, $P_i\cdot P_i=1$ for all $i\in[1,N]$. According to Lemma \ref{lem.equality}, there are two cases. And because $P$ is degenerative, which means there are at least two rows of $P$ are identical: $P_i\cdot P_j=1$ for some $i,j$. This is exactly the first case in the proof of Lemma \ref{lem.equality}. Because $0<\alpha<2$ and $\sigma_i> 1$ for $i\leq r$, we have:
    \begin{equation}
    \Gamma_{\alpha}=\sum_{i=1}^N\sigma_i^{\alpha}<\sum_{i=1}^N\sigma_i^2=N.
    \end{equation}

\end{proof}
\end{comment}
This theorem implies $\sum_i\sigma_i^{\alpha}$ is a better indicator for approximate dynamical reversibility than $\sum_i\sigma_i^2$ for $\alpha<2$ although both of them can achieve the maximized value $N$ when $P$ is reversible. However, if $P$ is degenerative, $\sum_i\sigma_i^2$ is also $N$, but $\sum_i\sigma_i^{\alpha}$ is not.

Next, we will discuss the minimum value and minimum point of $\Gamma_{\alpha}$.

\begin{comment}
\begin{cor}
    If $P_i\cdot P_i=1$ for all $i\in[1,N]$, then
    \begin{equation}
        \sqrt{P\cdot P^T}=P\cdot P^T
    \end{equation}
    
\end{cor}
\begin{proof}
    If the matrix $\sqrt{P\cdot P^T}$ can be written as a column vector $\sqrt{P\cdot P^T}=(Q_1,Q_2,\cdot\cdot\cdot,Q_N)^T$, where $Q_i$ is the row vector.
    
\end{proof}
\end{comment}
%for all $0\leq i,k\leq 1$
%thus:
%\begin{equation}
%\begin{aligned}   
%    \label{eqn.bounded}
%    \Gamma &=\Tr(\Sigma)= \Tr\left(\sqrt{P\cdot P^T}\right)=\Tr\left(\sqrt{\sum_{i=1}^N\sum_{j=1}^N P_i\cdot P_j}\right)\\
%    &\leq \sqrt{N^2}=N,
%\end{aligned}
%\end{equation}
%where, the equation holds when $P^T\cdot P=I$, that is $\chi$ is dynamically reversible.
%\end{proof}
\begin{lem}
    \label{lem:min_rank}
    The rank of a nonzero TPM $P$ can achieve the minimum value of $1$ if and only if the row vectors of $P_i$ for any $i$ are identical. In this case, 
    \begin{equation}
        \label{eq:gamma_identical_rows}
        \Gamma_{\alpha}=|P_1|^{\alpha}\cdot N^{\alpha/2}
    \end{equation}
    
\end{lem}
\begin{proof}
    If the rank of $P$ is 1, all the $N-1$ row vectors of $P_i, \forall i\in\{1,2,\cdots,N\}$ can be expressed by a linear function of the first row vector $P_1$, thus 
    \begin{equation}
        P_i=k\cdot P_1,
    \end{equation}
    with $k>0$. However, because $|P_i|_1=1$, thus $k$ must be one. Therefore:
    \begin{equation}
    \label{eq:all_rows_identical}
        P_i=P_j, \forall i,j\in\{1,2,\cdots,N\}.
    \end{equation}

    On the other hand, if Equation \ref{eq:all_rows_identical} holds, the rank of $P$ should be 1.

    In this case,
    \begin{equation}
        P\cdot P^T=|P_1|^2\cdot\mathbbm{1}_{N\times N},
    \end{equation}
    where $|\cdot|$ is the modulus for $\cdot$. Therefore, the eigenvalues of $P\cdot P^T$ are $(|P_1|^2\cdot N,0,\cdot\cdot\cdot,0)$. Thus, the singular values of $P$ should be $(|P_1|\cdot \sqrt{N},0,\cdot\cdot\cdot,0)$, this leads to Equation \ref{eq:gamma_identical_rows}
\end{proof}
\begin{lem}
\label{thm.lowerbound_gamma}
    For a given TPM $P=(P_1,P_2,\cdot\cdot\cdot,P_N)^T$, $\Gamma_{\alpha}=\sum_{i=1}^N\sigma_i^{\alpha}$ can reach its minimum 1 if and only if $P_i=\frac{1}{N}(1,1,\cdot\cdot\cdot,1)$ for $\forall i\in\{1,2,\cdots,N\}$.
\end{lem}
\begin{proof}
    When $P_i=\frac{1}{N}(1,1,\cdot\cdot\cdot,1)$ for $\forall i\in\{1,2,\cdots,N\}$, $|P_1|=N^{-1/2}$, and according to Lemma \ref{lem:min_rank}, 
    \begin{equation}
        \Gamma_{\alpha}=\sum_{i=1}^N\sigma_i^{\alpha}=|P_1|^{\alpha}\cdot N^{\alpha/2}=N^{-\alpha/2}\cdot N^{\alpha/2}=1
    \end{equation}
    for any $0\leq \alpha<2$.

    On the other hand, because the minimum value of $\sigma_i$ is zero, and $\Gamma_{\alpha}=\sum_i\sigma_i^{\alpha}\geq 0$, thus $\Gamma_{\alpha}$ can be minimized if the number of zero singular values is maximized. Notice that the number of non-zero singular values of $P$ is the same as the rank of $P$. Thus, the minimum of $\Gamma_{\alpha}$ can be reached when the minimized rank of $P$ is reached. In such case, according to Lemma \ref{lem:min_rank}, $\Gamma_{\alpha}=|P_1|^{\alpha}\cdot N^{\alpha/2}$, thus the minimized value of $\Gamma_{\alpha}$ is solely dependent on $|P_1|$. While, because $P_1$ is a probability distribution which satisfies $|P_i|_1=1$, thus, $P_1\cdot P_1$ can be minimized when all the elements are equal. Thus, 
    \begin{equation}
        P_1=\frac{1}{N}\cdot\mathbbm{1}.
    \end{equation}

    Therefore, $\frac{1}{N}\cdot \mathbbm{1}$ is the minimum point.
    
\end{proof}

Next, to illustrate why the dynamics reversibility measure $\Gamma_{\alpha}$ increases as the probability matrix $P$ asymptotically converges to a permutation matrix, or dynamically reversible one, we have the following lemmas and the theorem. {It is worth noting that the lemma and the theorem are well-established results in linear algebra. However, we provide our own proof for the sake of completeness and convenience.}
\begin{lem}
\label{lem:alpha_mean_function}
    The function $f(\alpha)=\left(\sum_{i=1}^Nx_i^{\alpha}\right)^{1/\alpha}$ is a monotonic decreasing function of $\alpha$ for any $x_i\geq 0, \forall i\in\{1,2,\cdots,N\}$ and $\alpha>0$.
\end{lem}
\begin{proof}
    Because:
    \begin{equation}
        \sum_{i=1}^N\left(x_i^{\alpha}\right)^2\leq \left(\sum_{i=1}^{N}x_i^{\alpha}\right)^2,
    \end{equation}
    thus:
    \begin{equation}
    \label{eq:inequal1}
    \log\frac{\sum_{i=1}^Nx_i^{2\alpha}}{\sum_{i=1}^{N}x_i^{\alpha}}\leq \log \sum_{i=1}^Nx_i^{\alpha}.
    \end{equation}
    Further, because $\log$ is a concave function, therefore:
    \begin{equation}
    \label{eq:inequal2}
        \sum_{i=1}^N\left(\frac{x_i^{\alpha}}{\sum_{j=1}^Nx_j^{\alpha}}\right)\cdot \log x_i^{\alpha}\leq \log\sum_{i=1}^N\left(\frac{x_i^{\alpha}}{\sum_{j=1}^Nx_j^{\alpha}}\cdot x_i^{\alpha}\right),
    \end{equation}
    thus, combining Eq. \ref{eq:inequal1} and \ref{eq:inequal2}, we have:
    \begin{equation}
    \label{eq:log_inequality}
        \frac{\sum_{i=1}^Nx_i^{\alpha}\cdot\log x_i^{\alpha}}{\sum_{i=1}^Nx_i^{\alpha}}\leq \log\sum_{i=1}^N x_i^{\alpha}.
    \end{equation}
    The equality holds when $x_i=0, \forall i\in\{1,2,\cdots,N\}$. Notice that the right hand term is $-\left(\frac{1}{\alpha}\right)'\log\sum_{i=1}^Nx_i^{\alpha}$, where $'$ represents the derivative with respect to $\alpha$, and the left hand term is $\frac{1}{\alpha}\cdot\left(\log\sum_{i=1}^{N}x_i^{\alpha}\right)'$, thus Equation \ref{eq:log_inequality} implies:
    \begin{equation}
        \left(\frac{1}{\alpha}\cdot\log\sum_{i=1}^Nx_i^{\alpha}\right)'=\frac{1}{\alpha}\cdot\left(\log\sum_{i=1}^{N}x_i^{\alpha}\right)'+\left(\frac{1}{\alpha}\right)'\log\sum_{i=1}^Nx_i^{\alpha} \leq 0.
    \end{equation}
    Therefore, $\log f(\alpha)=\frac{1}{\alpha}\cdot\left(\log \sum_{i=1}^Nx_i^{\alpha}\right)$ is a monotonic decreasing function of $\alpha$, and so does $f(\alpha)=\left(\sum_{i=1}^Nx_i^{\alpha}\right)^{1/\alpha}$.
\end{proof}
\begin{lem}
\label{thm:lowerbound}
    The approximate dynamical reversibility measure $\Gamma_{\alpha}$ is lower bounded by $||P||_{F}^{\alpha}$.
\end{lem}
\begin{proof}
    Because $0<\alpha<2$ and $\sigma_i\geq 0, \forall i\in\{1,2,\cdots,N\}$ but the equality can not hold for all $i\in\{1,2,\cdots,N\}$, this means $f(\alpha)=\left(\sum_{i=1}^N\sigma_i^{\alpha}\right)^{1/\alpha}$ is a strictly monotonic decreasing function of $\alpha$ according to Lemma \ref{lem:alpha_mean_function}, thus:
    \begin{equation}
        \Gamma_{\alpha}^{1/\alpha}=\left(\sum_{i=1}^N \sigma_i^{\alpha}\right)^{1/\alpha}\geq\left(\sum_{i=1}^N \sigma_i^2\right)^{1/2},
    \end{equation}
    the equality holds only when $\sigma_i=1$ or $0$ for all $i\in [1,N]$.
    According to Lemma \ref{lem.boundsize}:
    \begin{equation}
        \left(\sum_{i=1}^N \sigma_i^2\right)^{1/2}=\left(\sum_{i=1}^N P_i^2\right)^{1/2}=\left[\Tr(P\cdot P^T)\right]^{1/2}.
    \end{equation}
    While,
    \begin{equation}
        \sqrt{\sum_{i=1}^N P_i^2}=\sqrt{\sum_{i}^N\sum_j^{N}p_{ij}^2}=||P||_F
    \end{equation}
    Therefore:
    \begin{equation}
        \Gamma_{\alpha}\geq \left[\Tr(P\cdot P^T)\right]^{\alpha/2}=||P||_F^{\alpha}.
    \end{equation}
\end{proof}
As $||P||_F$ can achieve its maximum when $P_i$ is a one-hot vector, the lower bound of $\Gamma_{\alpha}$ increases as $P$ approaches a matrix with one-hot row vectors. This can be summarized as a following proposition:
\begin{lem}
\label{cor:gamma_increasing}
    The lower bound of the approximate dynamical reversibility measure $\Gamma_{\alpha}$ increases with the number of one hot vectors in the TPM $P$. 
\end{lem}
\begin{proof}
    Because $|P_i|_1=1$ for all $1\leq i\leq N$, therefore $||P_i||^2\leq 1$ and the equality holds if and only if $P_i$ is a one-hot vector. 
    Further, according to Lemma \ref{thm:lowerbound}, we have: 
    \begin{equation}
        \Gamma_{\alpha}\geq \left[\Tr{\left(P\cdot P^T\right)}\right]^{\alpha/2}=\left(\sum_{i=1}^N P_i^2\right)^{\alpha/2}.
    \end{equation}
    Thus, the lower bound $\left(\sum_{i=1}^N P_i^2\right)^{\alpha/2}$ increased as each row vector $P_i, \forall i\in\{1,2,\cdots,N\}$ converges to a one-hot vector. 
\end{proof}

However, $P$ may not be a permutation matrix because some row vectors may be similar, in which case $P_i$ for all $i\in [1,N]$ are not orthogonal to each other but collapse to one direction. Thus, we need to further prove a proposition to exclude this case.

\begin{comment}
\begin{lem}
    \label{thm.onehots}
    If $P$ is formed by one-hot vectors, and the rank is $r$, then $\Gamma_{\alpha}$ is $(N-r+1)^{\alpha/2}+r-1$, and this value increases with $r$.
\end{lem}

\end{comment}
\color{black}
\begin{lem}
 If $P$ is formed by one-hot vectors, and the rank of $P$ is $r$, so there are $r$ different row vectors, and each of them repeat $n_1,n_2,\ldots,n_r$ times $(n_1+\cdots+n_r=N)$, respectively, then $\Gamma_\alpha=\sqrt{n_1}^\alpha+\sqrt{n_2}^\alpha+\ldots\sqrt{n_r}^\alpha$. This value can reach its upper bound $r\sqrt{\frac{N}{r}}^\alpha$ when $n_1=n_2=\cdots=n_r=N/r$. This upper bound increases with $r$ if $N$ is fixed.
\end{lem}
\begin{proof}
   If $P$ is formed by one-hot vectors, and the rank of $P$ is $r$, then there are $r$ different row vectors. Suppose they are $e_{i_1},e_{i_2},\ldots,e_{i_r}$, and they repeat $n_1,\ldots,n_r$ times,respectively, where $n_1+n_2+\cdots+n_r=N$.
    
    Equivalently, there are $r$ nonzero columns and there are $n_1$ ones in the $i_1$th column, $n_2$ ones in the $i_2$th column, \dots, and $n_r$ ones in the $i_r$th column. Those ones lie in different rows. So $P^T \cdot P$ is a diagonal matrix with $n_1,n_2,\ldots,n_r$ and zeros as its diagonal elements. So the nonzero singular values of $P$ are $\sqrt{n_1},\sqrt{n_r},\ldots,\sqrt{n_r}$.
    \[\Gamma_\alpha=\sqrt{n_1}^\alpha+\sqrt{n_2}^\alpha+\ldots\sqrt{n_r}^\alpha\]

    For example,
    \[
    P = \begin{pmatrix}  
       1&0&0&0&0&0&0&0&0&0\\
       0&1&0&0&0&0&0&0&0&0\\
       0&0&1&0&0&0&0&0&0&0\\
       0&0&0&1&0&0&0&0&0&0\\
       0&0&0&1&0&0&0&0&0&0\\
       0&0&0&0&1&0&0&0&0&0\\
       0&0&0&0&1&0&0&0&0&0\\
       0&0&0&0&0&1&0&0&0&0\\
       0&0&0&0&0&1&0&0&0&0\\
       0&0&0&0&0&1&0&0&0&0
    \end{pmatrix}
\]  has $6$ different row vectors $e_1,e_2,e_3,e_4,e_5,e_6$. They repeat $1,1,1,2,2,3$ times, respectively. $P^T \cdot P= diag(1,1,1,2,2,3)$. The nonzero singular values of $P$ are $1,1,1,\sqrt{2},\sqrt{2},\sqrt{3}$.

The upper bound of $\Gamma_\alpha=\sqrt{n_1}^\alpha+\sqrt{n_2}^\alpha+\ldots\sqrt{n_r}^\alpha$ is $r\sqrt{\frac{N}{r}}^\alpha$ and it can be achieved if $N$ is a multiple of $r$ and $n_1=n_2=\cdots=n_r=N/r$. This is beacues $f(x)=x^\frac{\alpha}{2}$ is a concave function ($\alpha\in(0,2)$).
\[
\frac{1}{r}\sum_{i=1}^{r}f(n_i)\leq f(\frac{1}{r}\sum_{i=1}^{r}n_i)=f(N/r)
\]
So $\Gamma_\alpha=\sum_{i=1}^{r}f(n_i)\leq r\sqrt{\frac{N}{r}}^\alpha$.
The equality holds if and only if $n_1=n_2=\cdots=n_r=N/r$. For $\alpha<2$ and a fixed $N$, the upper bound $r\sqrt{\frac{N}{r}}^\alpha=r^{1-\frac{\alpha}{2}}N^\frac{\alpha}{2}$ increases with $r$.

\end{proof}

% In general, $\Gamma_\alpha$ tend to increase with $r$, as is shown in the figure.
% \begin{figure}[h]
%     \centering
%     \includegraphics[width=0.5\textwidth]{g1.eps}
%     \includegraphics[width=0.5\textwidth]{g05.eps}
%     \includegraphics[width=0.5\textwidth]{g15.eps}
%     \caption{$\Gamma_\alpha$ with $r$. Here $N=30$. All the possible matrices formed by one-hot row vectors are included.}
% \end{figure}
\color{black}

\subsection{Comparison between $\log \Gamma_{\alpha}$ and $EI$}
\label{sec:compare_gamma_EI_app}
In this subsection, we will establish the relationship between $\log\Gamma_{\alpha}$ and $EI$. First, we can synthesize the theoretical results about the minimum and maximum for both $\Gamma_{\alpha}$ and $EI$ to derive the following theorem:\\
\\
\textbf{Proposition \ref{thm:synthesize_theorem}}: \textit{For any TPM $P$ and $\alpha\in(0,1)$, both the logarithm of $\Gamma_{\alpha}$ and $EI$ share identical minimum value of $0$ and one common minimum point at $P=\frac{1}{N}\mathbbm{1}_{N\times N}$. They also exhibit the same maximum value of $\log N$ with maximum points corresponding to $P$ being a permutation matrix.}
\begin{proof}
    According to Lemma \ref{thm:EI_derivitive}, $EI$ has the minimum $0$ when $P$ has identical row vectors, and $P=\frac{1}{N}\cdot \mathbbm{1}_{N\times N}$ satisfies this condition. While, according to Lemma \ref{thm.lowerbound_gamma}, $\log\Gamma_{\alpha}$ has also the minimum value $0$ when $P=\frac{1}{N}\cdot \mathbbm{1}_{N\times N}$. Therefore, $EI$ and $\log\Gamma_{\alpha}$ share the same minimum.

    On the other hand, according to Lemma \ref{thm:EI_maximum}, $EI$ has the same maximum at $\log N$ when $P$ is a permutation matrix. So does $\log\Gamma_{\alpha}, \forall \alpha\in(0,2)$ according to Proposition \ref{thm.maximum}.

    Therefore, $EI$ and $\log\Gamma_{\alpha}$ share the same minimum and maximum for any $\alpha\in(0,2)$.
\end{proof}

We will prove a theorem that $EI$ is upper bounded by $\frac{2}{\alpha}\log\Gamma_{\alpha}$.\\
\\
\textbf{Theorem \ref{thm.EI_boundby_gamma}}. \textit{For any TPM $P$, its effective information $EI$ is upper bounded by $\frac{2}{\alpha}\log\Gamma_{\alpha}$, and lower bounded by $\log\Gamma_{\alpha}-\log N$.}

\begin{proof}
    First, because the upper bound of the average distribution $\bar{P}$ is:
    \begin{equation}
    \label{eq.HPbar_inequality}
        H(\bar{P})\leq \log N,
    \end{equation}
    the equality holds when $\bar{P}$ is $\frac{1}{N}\cdot \mathbbm{1}$. Second, according to the concavity of $\log$ function, we have:
    \begin{equation}
    \label{eq.HPi_inequality}
    \begin{aligned}
        -\frac{1}{N}\sum_{i=1}^N H(P_i)&=\frac{1}{N}\sum_{i=1}^N\sum_{j=1}^Np_{ij}\log p_{ij}\\
        &\leq \frac{1}{N}\sum_{i=1}^N\log \left(\sum_{j=1}^Np_{ij}^2\right)\\
        &\leq \log\sum_{i=1}^N\sum_{j=1}^N p_{ij}^2-\log N,
    \end{aligned}
    \end{equation}
    and the equality holds when $p_{ij}=1/N$ for all $1\leq i,j\leq N$. 
    Thus, we bring the two inequalities (Equation \ref{eq.HPbar_inequality} and Equation \ref{eq.HPi_inequality}) into Equation \ref{eq:twoterms_EI}, we obtain:
    \begin{equation}
    EI=-\frac{1}{N}\sum_{i=1}^N H(P_i) + H(\bar{P})\leq \log\sum_{i=1}^N\sum_{j=1}^N p_{ij}^2.
    \end{equation}
    This is the upper bound of $EI$ and the equality holds when $P=\frac{1}{N}\cdot \mathbbm{1}_{N\times N}$.
    
    On the other hand, according to Lemma \ref{thm.lowerbound_gamma},
    \begin{equation}
        \log\Gamma_{\alpha}=\log\sum_{i=1}^N\sigma_i^{\alpha}\geq \log ||P||_F^{\alpha}=\frac{\alpha}{2}\log\sum_{i=1}^N\sum_{j=1}^Np_{ij}^2,
    \end{equation}
    the equality holds when $\sigma_{i}=0,1, \forall i\in{1,2,\cdots,N}$ .
    Therefore:
    \begin{equation}
        EI\leq \frac{2}{\alpha}\log\Gamma_{\alpha}.
    \end{equation}
    
    Furthermore, because $EI\geq 0$ and $\Gamma_{\alpha}\leq N$, thus:
    \begin{equation}
        EI-\log\Gamma_{\alpha}\geq 0 - \log N=-\log N
    \end{equation}
    Therefore,
    \begin{equation}
        EI\geq \log\Gamma_{\alpha} - \log N,
    \end{equation}
    Thus, $EI$ is lower bounded by $\log\Gamma_{\alpha}-\log N$.
\end{proof}
Tighter bounds are expected to be found in future works.

\subsubsection{Quantification for Causal Emergence}
\label{sec:quantify_CE_app}
\begin{proposition}

    \label{thm:CE_bounds}
    For any given TPM $P$ with singular values $\sigma_1\geq \sigma_2\geq \cdots\geq \sigma_N$ and rank $r$, and for any given $\epsilon\in[0,\sigma_1]$, the degree of causal emergence of $P$ is:
    \begin{equation}
        \Delta\Gamma_{\alpha}=\frac{\sum_{i=1}^{r_{\epsilon}}\sigma_i^{\alpha}}{r_{\epsilon}}-\frac{\sum_{i=1}^N\sigma_i^{\alpha}}{N},
    \end{equation}
    and this degree satisfies the following inequality:
    \begin{equation}
    \label{eq:CE_bounds}
        0\leq \Delta\Gamma_{\alpha}\leq N-1,
    \end{equation}
    where $r_{\epsilon}=\max\{i|\sigma_i>\epsilon\}$. 
\end{proposition}
\begin{proof}
When $\epsilon\in(\sigma_N,\sigma_1]$, there is an integer $i\in(1,N]$ such that $\sigma_i>\epsilon$, and $r_{\epsilon}$ is the maximum one to satisfy this condition. Therefore: 
\begin{equation}
    \sigma_1\geq\sigma_2\geq\cdots\geq\sigma_{r_{\epsilon}}>\epsilon\geq\sigma_{r_{\epsilon}+1}\geq\cdots\geq\sigma_N\geq 0,
\end{equation}
thus: 
\begin{equation}
    \sum_{i=1}^{r_{\epsilon}}\sigma_i^{\alpha}> r_{\epsilon}\cdot \epsilon^{\alpha},
\end{equation}
and:
\begin{equation}
    -\sum_{i=r_{\epsilon}+1}^{N}\sigma_i^{\alpha}\geq -(N-r_{\epsilon})\cdot \epsilon^{\alpha}.
\end{equation}
Therefore:
\begin{equation}
\label{eq:CE_lower}
    \frac{\sum_{i=1}^{r_{\epsilon}}\sigma_i^{\alpha}}{r_{\epsilon}}-\frac{\sum_{i=1}^N\sigma_i^{\alpha}}{N}=\sum_{i=1}^{r_{\epsilon}}\sigma_i^{\alpha}(\frac{1}{r_{\epsilon}}-\frac{1}{N})-\frac{\sum_{i=r_{\epsilon}+1}^{N}\sigma_i^{\alpha}}{N}> r_{\epsilon}\epsilon^{\alpha}(\frac{1}{r_{\epsilon}}-\frac{1}{N})-\frac{(N-r_{\epsilon})\cdot \epsilon^{\alpha}}{N}= 0.
\end{equation}
While, when $\epsilon<\sigma_N$, $r_{\epsilon}=N$, so $\Delta\Gamma_{\alpha}=0$, therefore:
\begin{equation}
    \Delta\Gamma_{\alpha}\geq 0
\end{equation}

Further, when $\epsilon=\sigma_N$, $r_{\epsilon}=N$, thus $\Delta\Gamma_{\alpha}=0$.
\\
On the other hand, because $r_{\epsilon}\leq N$, and according to Proposition \ref{thm.maximum}, thus:
\begin{equation}
    \sum_{i=1}^{r_{\epsilon}}\sigma_i^{\alpha}\leq\sum_{i=1}^{N}\sigma_i^{\alpha}\leq N.
\end{equation}

Since $1\leq r_{\epsilon}\leq N$, therefore 
\begin{equation}
    0\leq \frac{1}{r_{\epsilon}}-\frac{1}{N}\leq 1-\frac{1}{N}.
\end{equation}
Thus, 
\begin{equation}
\label{eq:CE_upper}
    \frac{\sum_{i=1}^{r_{\epsilon}}\sigma_i^{\alpha}}{r_{\epsilon}}-\frac{\sum_{i=1}^N\sigma_i^{\alpha}}{N}\leq \sum_{i=1}^N\sigma_i^{\alpha}(\frac{1}{r_{\epsilon}}-\frac{1}{N})\leq N-1.
\end{equation}
Combining Equation \ref{eq:CE_lower} and \ref{eq:CE_upper}, we obtain Equation \ref{eq:CE_bounds}
\end{proof}

\begin{cor}
\label{eq:CE_occurence_condition}
    For any given TPM $P$ with singular values $\sigma_1\geq \sigma_2\geq \cdots\geq \sigma_N$ and rank $r$, according to Definition \ref{dfn:vague_emergence}, causal emergence occurs if and only if $\Delta\Gamma_{\alpha}(\epsilon)>0$ for some $\epsilon\geq 0$.
\end{cor}
\begin{proof}

    \textbf{Case I}: When $\epsilon>0$, according to Definition \ref{dfn:vague_emergence}, if there exists an integer $i\in\{1,2,\cdots,N\}$ such that $\sigma_i>\epsilon$ for any $\epsilon\in[0,\sigma_1]$, that is $\epsilon>\sigma_N$, then the vague CE occurs. Then, according to Proposition \ref{thm:CE_bounds}, $\Delta\Gamma_{\alpha}(\epsilon)>0$ in this case.

    Otherwise, if $\Delta\Gamma_{\alpha}(\epsilon)>0$, then $r_{\epsilon}<N$, where $r_{\epsilon}=\max\{i|\sigma_i>\epsilon\}$, and therefore $\sigma_{r_{\epsilon}}>\epsilon$. According to Definition \ref{dfn:vague_emergence}, vague CE occurs.

    \textbf{Case II}: When $\epsilon=0$, according to Definition \ref{dfn:clear_emergence}, the rank $r$ of $P$ is less than $N$ and $\sigma_j=0, \forall j\in\{r+1,r+2,\cdots,N\}$, therefore 
    \begin{equation}
        \frac{\sum_{i=1}^r\sigma_i^{\alpha}}{r}>\frac{\sum_{i=1}^N\sigma_i^{\alpha}}{N},
    \end{equation}
    so $\Delta\Gamma_{\alpha}>0$.
    
    Otherwise, if $\Delta\Gamma_{\alpha}>0$, there exists $r_{\epsilon=0}<N$ such that $r_{\epsilon=0}=\max\{i|\sigma_i>0\}$. Thus, $\sigma_i=0, \forall i>r_{\epsilon}$, that is $r_{\epsilon=0}$ is the rank of the matrix $P$. Therefore, clear CE occurs according to Definition \ref{dfn:clear_emergence}.

\end{proof}
\nocite{*}
% \color{red}
\section{\label{sec:causal_emergence_reason}The Relationship between SVD and Max EI}

In this paper, we define clear and vague causal emergence using the spectrum of singular values, as outlined in Definitions \ref{dfn:clear_emergence} and \ref{dfn:vague_emergence}. The rationale behind these definitions is that, approximately, a necessary condition for maximizing Effective Information (EI) is that the optimal coarse-graining strategy should allocate more probability mass to the directions of singular vectors corresponding to the largest singular values. However, this condition is not exact, as the actual coarse-graining strategy must meet additional requirements—for instance, the coarse-grained TPM for macro-dynamics must satisfy the normalization condition for each row vector, and the grouping of micro-states must be well-defined and deterministic.

In this section, we theoretically demonstrate why the necessary condition should be satisfied approximately by examining one of the simplest case, the coarse-graining strategies can be represented by a clustering matrix and the coarse-grained TPM can be calculated by the multiplication of the matrices. Subsequently, we provide two examples to illustrate when will consistent conclusion about CE can be derived.

\subsection{Theoretical Analysis for the Necessary Condition of EI Maximization}
In the framework of Erik Hoel's theory of causal emergence, the occurence of causal emergence for a given Markovian dynamical systems relies on the coarse-graining strategy which maximizing the EI of the macro-dynamics after the coarse-graining. 

We define a coarse-graining method using an $N\times r$ clustering matrix $\Phi=(\Phi_1^T,\Phi_2^T,\cdots,\Phi_r^T)$, where each vector $\Phi_i$ indicates the membership in the $i$th cluster. Specifically, $\Phi_{j,i}=1$ if the $j$th micro-state belongs to the $i$th macro-state.

Since each microstate corresponds to a unique macrostate, all vectors $\Phi_i$ for $i\in\{1,2,\cdots,r\}$ are mutually orthogonal. Thus, the transition probability matrix (TPM) for the coarse-grained macro-dynamics can be expressed as follows:

\begin{equation}
    \label{eqn:exact_coarse-grained_macro}
    P'=D\cdot\Phi^T\cdot P\cdot \Phi.
\end{equation}
where $D=diag(1/\sum_{j=1}^N\Phi_{1,j},1/\sum_{j=1}^N\Phi_{2,j},\cdots,1/\sum_{j=1}^N\Phi_{r,j})$ is the normalization coefficients such that $P'$ are composed by normalized row probability vectors. Equation \ref{eqn:exact_coarse-grained_macro} illustrates a naive coarse-graining method that collapses micro-states into macro-states by summing the collapsed probabilities across different columns and averaging across all rows in $P$.

While the expression is exact, its asymmetric form limits further theoretical analysis. To enable this, we normalize the clustering vectors $\Phi_i$ by dividing each by its norm, resulting in $\phi_i=\Phi_i/|\Phi_i|$ and $\phi=(\phi_1^T,\phi_2^T,\cdots,\phi_r^T)$. This allows us to approximately transform Equation \ref{eqn:exact_coarse-grained_macro} into a symmetric form:

\begin{equation}
    \label{eqn:coarse-grained_macro}
    P'\approx \phi^T\cdot P\cdot \phi,
\end{equation}

While, by singular value decomposition, $P$ can be written as,
\begin{equation}
    \label{eqn:singular_value_decomposition}
    P=\sum_{i=1}^N\sigma_i U_i V_i^T,
\end{equation}
where $U_i$(with size $N\times 1$) and $V_i^T$ (with size $1\times N$ $\forall i\in\{1,2,\cdots,N\}$ are singular vectors corresponding to the $i$th largest singular value. Thus,
\begin{equation}
\label{eqn:square_P}
P\cdot P^T=\sum_{i=1}^N\sigma_i^2\cdot U_i\otimes U_i^T,
\end{equation}
where $\otimes$ is the outer product. Therefore, inserting Equation \ref{eqn:square_P} into Equation \ref{eqn:coarse-grained_macro}, we have:
\begin{equation}
    \label{eqn:coarse_decomposition}
    P'\cdot P'^T\approx\phi^T\cdot P\cdot\phi\cdot\phi^T \cdot P^T\cdot\phi=\phi^T\cdot P\cdot P^T \cdot\phi=\sum_{i=1}^N\sigma_i^2\cdot (\phi^T\cdot U_i) \otimes (U_i^T\cdot\phi)
\end{equation}

This equation holds because $\Phi_i,\forall i\in\{1,2,\cdots,r\}$ are orthogonal each other, and as a result $\phi\cdot\phi^T=I_{N\times N}$. Therefore, according to Lemma \ref{thm:lowerbound}, the $1/\alpha$-ordered power of the approximate dynamical reversibility of $P'$ satisfies
\begin{equation}
    \label{eqn:P_prime}
    \Gamma_{\alpha}^{1/\alpha}(P')\geq ||P'||_F=\mathrm{Tr}(P'\cdot P'^T)\gtrapprox \sum_{i=1}^{r}\sigma_i^2\cdot \mathrm{Tr}(W_i\otimes W_i^T)=\sum_{i=1}^r\sum_{j=1}^r \sigma_i^2\cdot(\phi_{j}\cdot U_i)^2,
\end{equation}
where, $W_i\equiv \phi^T\cdot U_i=(\phi_1\cdot U_i,\phi_2\cdot U_i,\cdots,\phi_r\cdot U_i)^T$, and the second inequality holds because the smallest $N-r$ singular values are cut-off. Thus, if $\phi=(\phi_1^T,\phi_2^T,\cdots,\phi_r^T)$ is selected such that there is at least one vector, say $\phi_j^T$, being parallel with the singular vector $U_i^T$, such that $||W_i||^2$ could be maximized for $\forall i\leq r$. This is possible since both $\phi$ and $U$ are all mutually orthogonal. This is equivalent to assign more probability mass on the directions of $U_i^T$ for largest singular values.

As a result, $\Gamma_{\alpha}(P')$ could be maximized. While, according to the approximate relationship $\log \Gamma_{\alpha}\sim EI$, $EI$ could also be maximized.
Therefore, to maximize $EI$, we should select the coarse-graining strategy that can assign more probability mass on the directions of singular vectors corresponding to the largest singular values.

However, in real applications, the coarse-graining strategy $\Phi$ has some constraints (e.g., it should be hard grouping method, and the final TPM of the macro-dynamics should satisfy normalization condition) such that the $\phi_i$ can not parallel the singular vectors exactly. Thus, this requirement is only an approximate necessary condition for maximizing EI.
\subsection{A Numeric Example for Clear Emergence}
We will show the consistency between the SVD and the EI maximization methods with two examples. The first example is shown in Figure \ref{fig:examples}(c) and (d). The TPM is:
\begin{equation}
    P=\begin{pmatrix}
        1/3 & 1/3 & 1/3 & 0\\
        1/3 & 1/3 & 1/3 & 0\\
        1/3 & 1/3 & 1/3 & 0\\
        0 & 0 & 0 & 1\\
    \end{pmatrix}
\end{equation}

Because $rank(P)=2$, the clear CE occurs. $EI(P)=0.81$. We select $r=2$, and the optimal strategy of coarse-graining for maximization of $EI$ can be written as,
\begin{equation}
    \Phi=\begin{pmatrix}
        1,0\\
        1,0\\
        1,0\\
        0,1
    \end{pmatrix},
\end{equation}

and the corresponding $\phi$ being the normalization of $\Phi$ is:
\begin{equation}
\label{eqn:phi_matrix_example}
    \phi=\begin{pmatrix}
        \frac{1}{\sqrt{3}},0\\
        \frac{1}{\sqrt{3}},0\\
        \frac{1}{\sqrt{3}},0\\
        0,1\\
    \end{pmatrix}.
\end{equation}

The left matrix composed by all the singular vectors of $P$ is:
\begin{equation}
\label{eqn:Umatrix_example}
    U=\begin{pmatrix}
        0&\frac{1}{\sqrt{3}}&-\frac{1}{\sqrt{2}}&-\frac{1}{\sqrt{6}}\\
        0&\frac{1}{\sqrt{3}}&0&\sqrt{\frac{2}{3}}\\
        0&\frac{1}{\sqrt{3}}&\frac{1}{\sqrt{2}}&-\frac{1}{\sqrt{6}}\\
        1&0&0&0
    \end{pmatrix}.
\end{equation}

Comparing Equations \ref{eqn:phi_matrix_example} and \ref{eqn:Umatrix_example}, we know the first column vectors in Equation \ref{eqn:phi_matrix_example} are parallel the second and the first column vectors in Equation \ref{eqn:Umatrix_example}, and the inner products between $U_3^T$ and $U_4^T$ and $\phi$ are zeros. The final coarse-grained TPM is:
\begin{equation}
    P'=\begin{pmatrix}
        1&0\\
        0&1\\
    \end{pmatrix}
\end{equation}

Therefore, $EI(P')=1$ which is larger than $0.81$, and $CE=1-0.81=0.19$. This consistent with the conclusion drawn by SVD method where the degree of CE is $\Delta\Gamma=1/2$.

\begin{subsection}{An Example for Vague Causal Emergence}
We can also give an example of vague CE to illustrate the unreasonable nature of the optimal coarse-graining method for maximizing EI. This example is original shown in \cite{eberhardt2022causal} to demonstrate the weakness of Erik Hoel's theory of causal emergence.

Consider a Markov chain with 3 states, and its TPM is:
\begin{equation}
    P=\begin{pmatrix}
        0.3&0.6&0.1\\
        0.6&0.2&0.2\\
        0&0&1\\
    \end{pmatrix}
\end{equation}

Because $rank(P)=3$, there is no clear CE. Its singular values are:$\sigma_1=1.06,\sigma_2=0.80,\sigma_3=0.35$. $EI(P)=0.66$.
We can set the vagueness $\epsilon=0.35$, such that the vague CE occurs with the degree of $\Delta\Gamma=\frac{1.06+0.8}{2}-\frac{1.06+0.8+0.35}{3}=0.19$. 

The $U$ matrix with singular vectors is:

\begin{equation}
    U=\begin{pmatrix}
        0.32&-0.67&-0.67\\
        0.40&-0.55&0.74\\
        0.86&0.50&-0.09\\
    \end{pmatrix}
\end{equation}

There are totally three possible reasonable coarse-graining strategies: clustering two micro-states as a group and leaving the other micro-state itself as a group, namely, $\{\{1,2\},3\},\{\{1,3\},2\},\{\{2,3\},1\}$. The represented vectors for these three strategies, as well as the singular vectors of $U$ can be visualized by Figure \ref{fig:vectors}.

\begin{figure}
    \centering
    \includegraphics[width=0.5\linewidth]{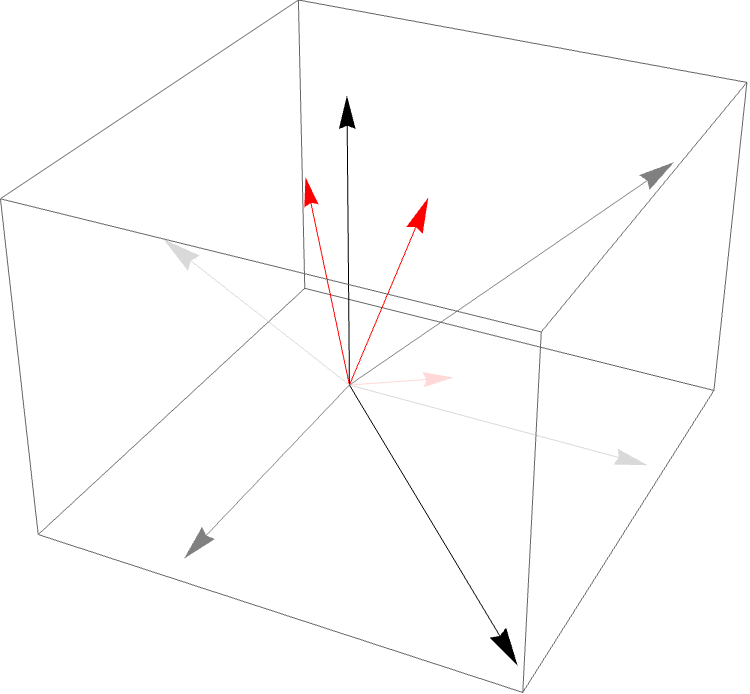}
    \caption{The singular vectors and the vectors representing the coarse-graining strategies. The red arrows are the singular vectors in $U$(the magnitude is multiplied by the corresponding singular values, the light red arrow represents the singular vector with the smallest singular value), the black arrows represent the optimal coarse-graining strategy with EI maximization, $\{\{1,2\},3\}$. The gray arrows represent the coarse-graining strategy $\{\{2,3\},1\}$, and the light gray arrows represent the strategy $\{\{1,3\},2\}$.}
    \label{fig:vectors}
\end{figure}

It is clearly to see in Figure \ref{fig:vectors} that only the black arrows representing the optimal strategy are parallel with the red arrows which representing the major singular vectors corresponding to the largest two singular values. All other vectors representing the other strategies for coarse-graining all have the non-zero projections on the direction represented by the light red arrow, which is unreasonable.

The optimal strategy is
\begin{equation}
    \Phi=\begin{pmatrix}
        1&0\\
        1&0\\
        0&1\\
    \end{pmatrix}
\end{equation}

Finally, the optimal macro-level coarse-grained TPM is:
\begin{equation}
    P'=\begin{pmatrix}
        0.85&0.15\\
        0&1\\
    \end{pmatrix}
\end{equation}
It is EI turns out to be 0.68, and the $CE=0.68-0.66=0.02$. Thus, the CE occurs according to EI maximization. 

However, as pointed out in reference \cite{eberhardt2022causal}, the optimal coarse-graining strategy of EI maximization is unreasonable because it combines the states with dissimilar vectors (the first and the second row vectors), this issue is referred to as ambiguity in \cite{eberhardt2022causal}, as it introduces uncertainty when reducing the intervention from macro-states to micro-states. Also, the coarse-grained TPM $P'$ deviates from $P$, this can also be observed by the relatively large $\epsilon$. 

Another serious problem is the non-commutativity between abstraction(coarse-graining) and marginalization(time evolution) in this example because:

\begin{equation}
    P\cdot\Phi=\begin{pmatrix}
        0.9&0.1\\
        0.8&0.2\\
        0&1
    \end{pmatrix},
\end{equation}

but:
\begin{equation}
    \Phi\cdot P'=\begin{pmatrix}
        0.85& 0.15\\
        0.85& 0.15\\
        0.& 1.
    \end{pmatrix}.
\end{equation}

Thus,
\begin{equation}
    P\cdot \Phi\neq \Phi\cdot P'
\end{equation}

Therefore, the coarse-graining operator does not commute with the time evolution operator. That indicates that the optimal coarse-graining method is unreasonable on this example although it can maximize EI.

\end{subsection}

\subsection{Consistency of max EI and SVD Methods in Experiments with Cellular Automata and Complex Networks}
\label{sec:consistency_maxEI_SVD}

This section experimentally demonstrates the relationship between the direction of the clustering method in coarse-graining strategy that maximizes EI, represented by vectors, and the directions of the singular vectors of the TPM, as illustrated in \ref{fig:greedyei_eg}. The experimental subjects are cellular automata and stochastic block models mentioned in Figures \ref{fig:CA} and \ref{fig:emergence_renormalization} in the main text for reference. Coarse-graining strategies consist of two steps: 1. identifying the optimal state clustering method using a greedy approach as outlined in \cite{Hoel2013,klein2020emergence}; 2. constructing the coarse-grained TPM by summing rows and averaging columns in $P$ based on the clustering method from the first step.

\begin{figure}
    \centering
    \includegraphics[width=1\linewidth]{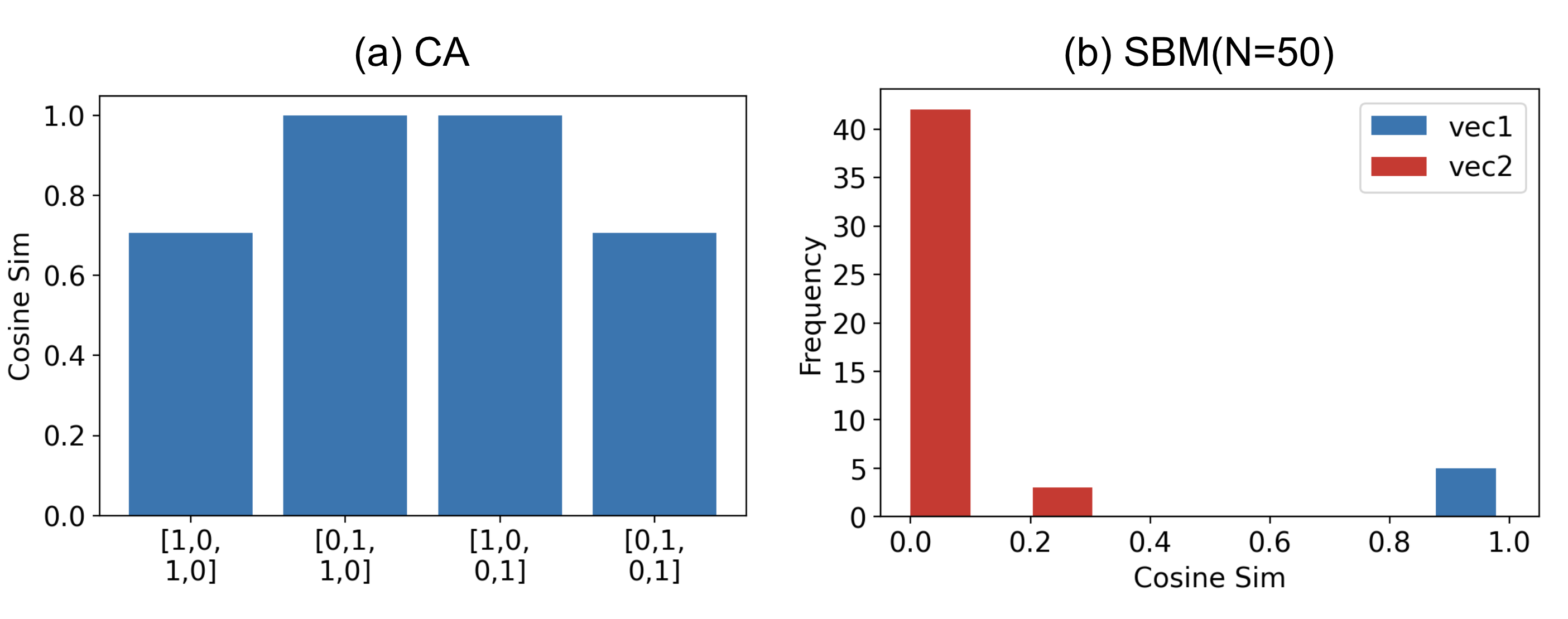}
    \caption{The cosine similarities between state clustering methods in coarse-graining strategies, represented by vectors that maximize EI, and the singular vectors of the TPM for cellular automata (a) and the stochastic block model of complex networks (b) are analyzed. In Figure (a), the four coordinates on the horizontal axis represent the four types of local TPMs (the only possible 2x2 permutation matrices) for the cellular automaton discussed in Figure \ref{fig:CA} in the main text. For each local TPM, we identify the coarse-graining strategy—methods for clustering states—represented by vectors that maximize EI using a greedy algorithm (\cite{Hoel2013,klein2020emergence}), and then calculate the cosine similarity (the vertical axis) between these vectors and the singular vectors corresponding to the largest $r$ singular values, with $r$ determined by the EI maximization results. Figure (b) illustrates the distribution of cosine similarities between the coarse-graining strategies (methods for clustering nodes into communities) represented by vectors for maximizing EI and the singular vectors for the largest (vec1) and smallest (vec2) singular values for the complex networks randomly generated by stochastic block models (SBM) with 50 nodes, 5 blocks(communities), and parameters $p = 1$ (connection probability within communities) and $q = 0$ (connection probability between communities). We utilize the normalized EI, Eff, as our objective function to compare networks of varying sizes. The horizontal axis represents cosine similarities, while the vertical axis indicates the frequency of these similarities within specific intervals.}
    \label{fig:greedyei_eg}
\end{figure}

Both figures show that the directions of the vectors representing the states clustering method in coarse-graining strategy that maximizing EI align closely with the singular vectors of the largest retained singular values. Figure \ref{fig:greedyei_eg}(a) shows the cosine similarities greater than 0.5 between the clustering method vectors and the singular vectors corresponding to the largest singular values. Figure \ref{fig:greedyei_eg}(b) illustrates the distribution of cosine similarities between the clustering method represented with vectors that maximize EI and the singular vectors, where vec1 (blue column) represents the vectors corresponding to the largest singular values, and vec2 (red columns) represents the vectors corresponding to the smallest values. The similarities for vec1 cluster between 0.8 and 1, while the similarities for vec2 are close to 0. This strongly indicates a high positive correlation between the optimal state clustering directions for EI maximization and the singular vectors of the singular vectors with largest singular values, along with a clear dissimilarity to the ones for smallest singular values.

% \color{black}
\section{\label{sec:comparison_experiments}Experiments on Testing the Correlation between $EI$ and $\Gamma$}
In this section of the Supplementary, we will introduce the details of our experiments on the relationship between the approximate dynamical reversibility $\Gamma$ and $EI$ on various generated TPMs. The generative model of TPMs have three classes: softening of permutation matrix, softening of controlled degenerative TPMs, and random normalized. 
\subsection{\label{sec:perturbation_permutation}Softening of Permutation Matrix}
In this series experiments, we will find out what the relationships between $\Gamma$ and $EI$ are on a variety of TPMs with different deviations from the reversible TPMs (permutation matrix) and different sizes.

For given size $N$, The TPM is generated by three steps: 1). Randomly sample a permutation matrix with dimension $N\times N$; 2). For each row vector $P_i$ in $P$, suppose the position of the 1 element is $j_i$, we fill out all entries of $P_i$ with the probabilities of a Gaussion distribution center at $j_i$, that is, $p'_{i,j}=\frac{1}{\sqrt{2\pi}\sigma}\exp{-\frac{(j-j_i)^2}{\sigma^2}}$, where, $\sigma$ is a free parameter for the degree of softening; 3). Normalize this new row vector $P'_i$ by dividing by $\sum_{j=1}^Np'_{ij}=1$, such that the modified matrix $P'$ is also a TPM.

In this model, the unique parameter $\sigma$ can control the degree of deviations from the original TPM. When $\sigma=0$, we recover the original TPM. And when $\sigma$ increases to very large value, then the row vectors converge to the vector $\mathbbm{1}/N$. \ref{fig:perturbed_TPM} shows the TPMs before and after the update on $\sigma=10$, where the colors represent the probabilities.

\begin{figure}
    \centering
    \includegraphics[width=0.7\linewidth]{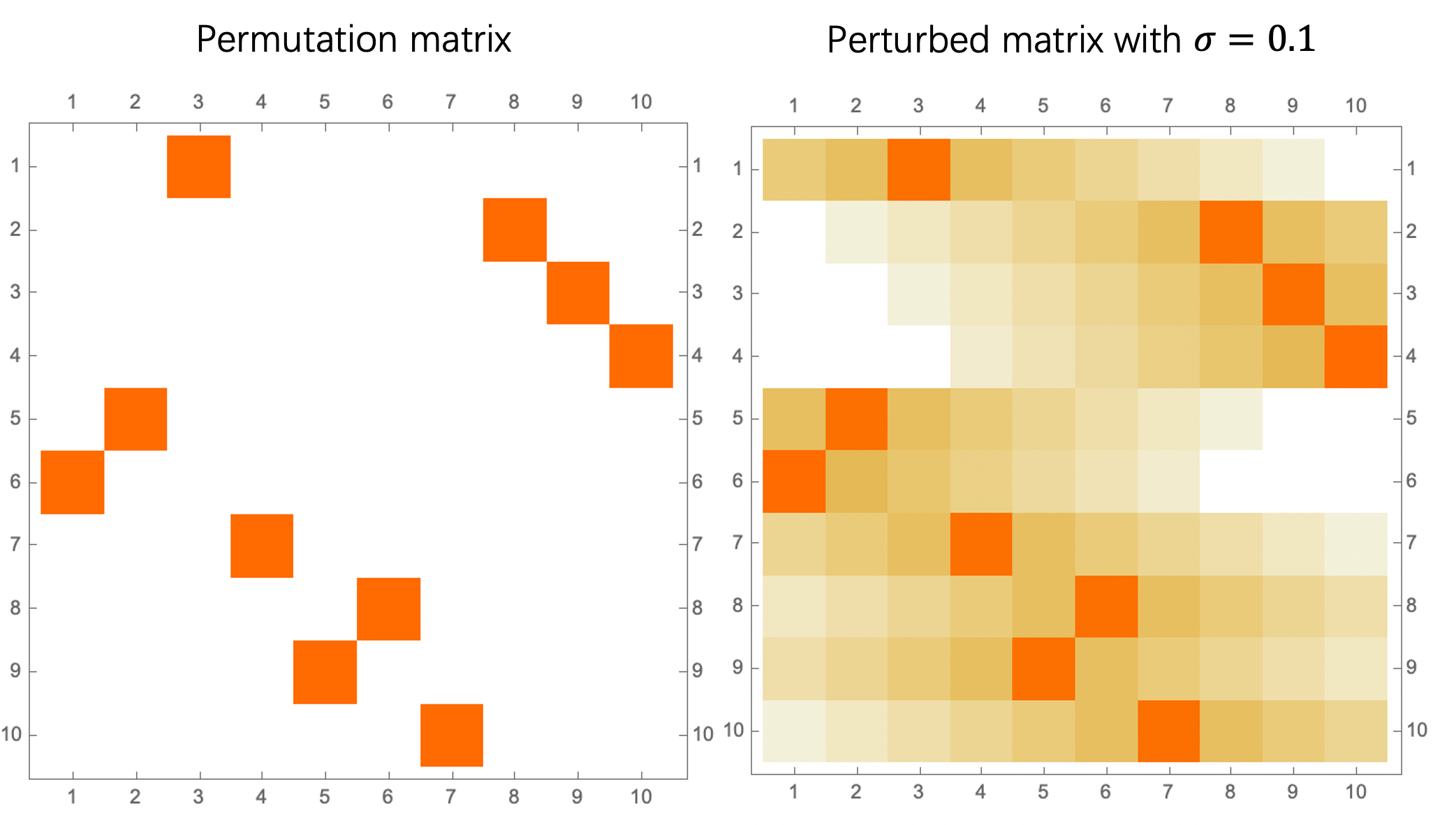}
    \caption{The original TPM which is a permutation matrix and the perturbed matrix after softening.}
    \label{fig:perturbed_TPM}
\end{figure}

By adjusting different $\alpha$, we can obtain the similar curves between $\Gamma_{\alpha}$ and $EI$ as shown in \ref{fig:alpha}. We can observe that the positive correlations between $EI$ and $\Gamma_{\alpha}$ can not be changed by different $\alpha$, the shapes of the curves are different. 
\begin{figure}
    \centering
    \includegraphics[width=1\linewidth]{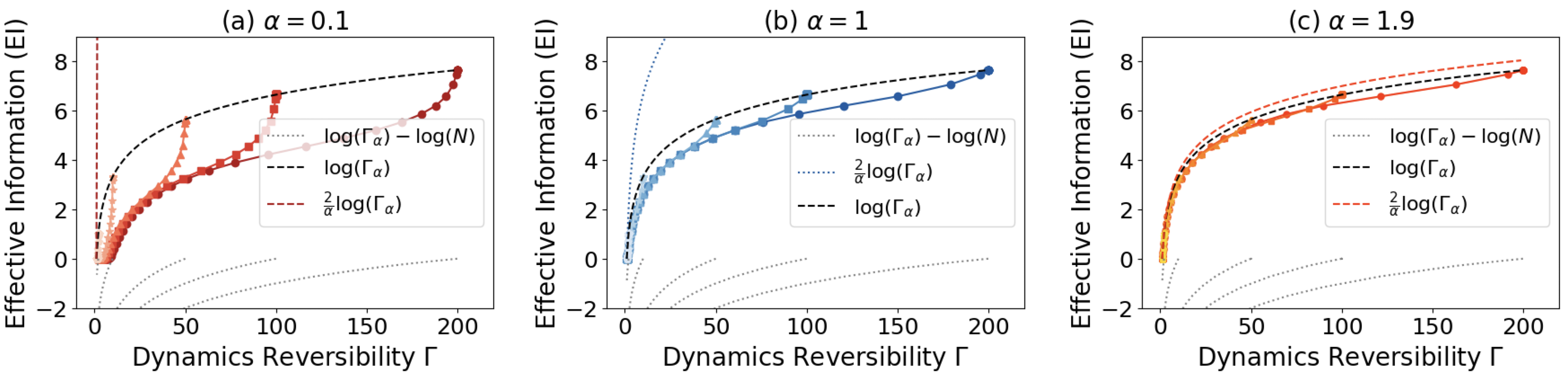}
    \caption{The relationships between $EI$ and $\Gamma_{\alpha}$ for different $\alpha$. The theoretical and empirical upper bounds are shown as red and black dashed lines, while the theoretical lower bounds which change with $N$ are also shown as dotted lines.}
    \label{fig:alpha}
\end{figure}

\subsection{\label{sec:degeneracy}Softening of Controlled Degeneracy}

The second model is very similar to the first one, however, the original matrix is not a permutation matrix, but a degenerated matrix. Here, a TPM is degenerative means that there are some row vectors are identical, and the number of identical row vectors is denoted as $N-r$ which is the controlled variable, where $r$ is the rank of $P$. By tuning $N-r$, we can control the degeneracy~\cite{Hoel2013} of the $TPM$ as \ref{fig:degenerate_TPM} shows.

\begin{figure}
    \centering
    \includegraphics[width=1\linewidth]{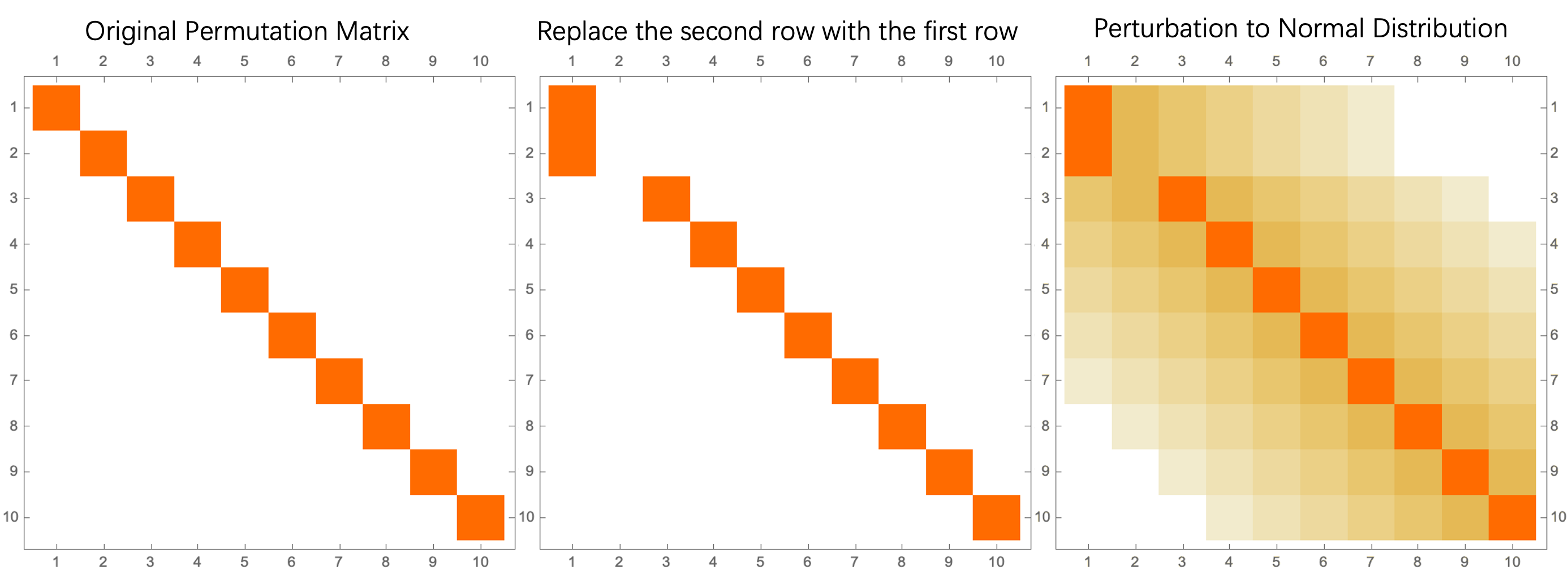}
    \caption{The original TPM ($N=10$ identity matrix), the replaced degenerated one by controlling $N-r=2$, and the perturbed TPM of the replaced one by softening with $\sigma=0.1$. Different colors represent different values of probability.}
    \label{fig:degenerate_TPM}
\end{figure}

Without losing generality, we can start from an identity matrix, and change the controlled $N-r$ row vectors into the same one-hot vectors with all elements 0 except the first one is 1. After that, we soften the one hot vectors with the same method mentioned in Section \ref{sec:perturbation_permutation} to obtain the results in the main text and Figure \ref{fig:comparison}(b).

\subsection{\label{sec:random_normalization}Random Normalization}
In this model, only two steps are required to generate a TPM: 1) Sample a row random vector from a uniform distribution in $[0,1]$, 2) normalize this row vector such that the generated matrix is a TPM.

%\begin{figure}
%    \centering
%    \includegraphics[width=1\linewidth]{figures/degeneracy.png}
%    \caption{(a). Degeneracy of different disturbance change as $\alpha$. Here, we choose $\sigma=0.001,1000$, 1/2 and 1/3 quantiles under logarithmic division, that is $\sigma=0.0316,0.1$ to display. (b). EI versus $\Gamma$ for different degeneracy. Each point here is the average result of a different certainty under different disturbance noise $\sigma$.}
%    \label{fig:degenerate_results}
%\end{figure}

\section{\label{sec:simplist_case}Analytic Solutions on the Parameterized 2*2 TPM}
We compare $EI$ and $\Gamma$ on a simplest example, the TPM is as:
\begin{equation}
\label{eq:pq_matrix}
    P=\begin{pmatrix}p & 1-p \\1-q & q\end{pmatrix},
\end{equation}
where $p$ and $q$ are all free parameters in the range of $[0,1]$. With this TPM, we can explicitly write down the expression for $EI$:
\begin{equation}
\label{eq:pq_matrix_EI}
\begin{aligned}
    EI=&\frac{1}{2}\left[p\log_2\frac{2p}{1+p-q}+(1-p)\log_2\frac{2(1-p)}{1-p+q}\right.\\
    &+\left.(1-q)\log_2\frac{2(1-q)}{1+p-q}+q\log_2\frac{2q}{1-p+q}\right],
\end{aligned}
\end{equation}

and $\Gamma$:
\begin{equation}
\label{eq:pq_matrix_gamma}
\begin{aligned}
    \Gamma=\sqrt{p^2+(1-p)^2+(1-q)^2+q^2+2|1-q-p|}.
\end{aligned}
\end{equation}
According these two expressions, we plot the landscapes in Figure \ref{fig:comparison}(e) and (f). 

\section{\label{sec:other_example}Applying Our Coarse-graining Method on More Examples}
\subsection{\label{sec:example_booleans}Examples of Boolean Networks in Hoel(2003)}
To compare our methods and EI maximization method on quantification of causal emergence and coarse-graining a Markov chain, we apply our method to the examples of causal emergence in Hoel et al's original papers \cite{Hoel2013} and \cite{Hoel2017}. All the examples show clear causal emergence. And almost identical reduced TPMs are obtained for all the examples.

%In Figure \ref{fig:SI_ex1}, we present the results for the same example Hoel mentioned, which uses a directed Boolean network of 6 nodes. Our findings align perfectly with Hoel's, proving that the equivalence of the two theoretic frameworks. 

Figures \ref{fig:SI_ex2} and \ref{fig:SI_ex3} show additional examples from the Supplementary, where our coarse-graining method obtain a coarse Markov chain with higher EI compared to Hoel's findings. The main difference lies in: Hoel groups node variables to form a new macro Boolean network and then calculates EI for the TPM. We, on the other hand, group states directly, which might result in a TPM that doesn't fully represent the original Boolean network in terms of variables. Therefore, the extension of our method to variable-based but not state-based coarse-graining method is deserve future studies.

\newpage

\begin{figure}[h!]
    \centering
    \includegraphics[width=1\linewidth]{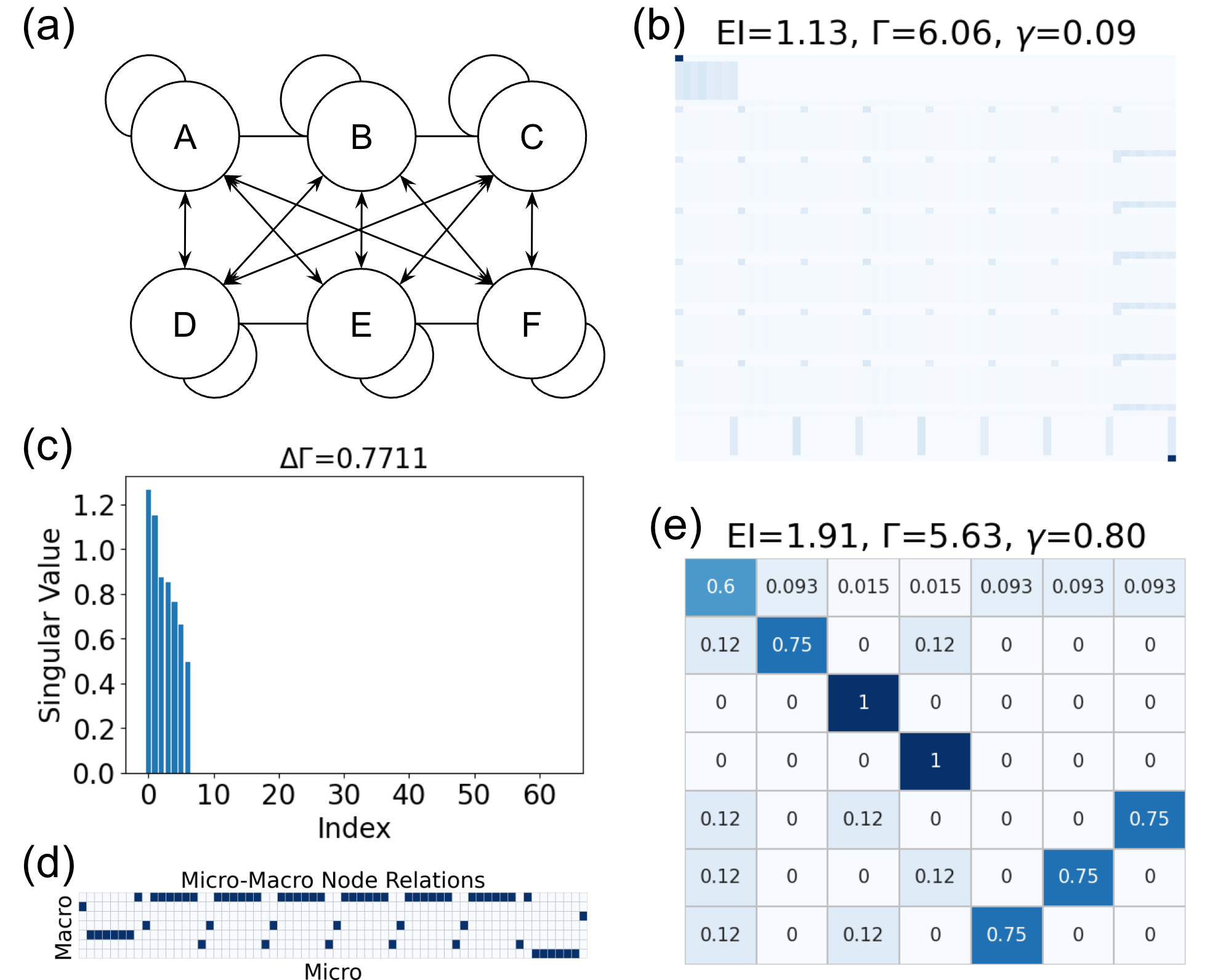}
    \caption{(a) represents a fully connected Boolean network with two distinct types of edges. (b) depicts the Transition Probability Matrix (TPM) as illustrated in Fig. S2 of \cite{Hoel2013}. (c) displays the singular spectrum of (b). (d) showcases the projection matrix. (e) illustrates the coarse-grained TPM derived from (b) using our coarse-graining approach. In our analysis, we observe a higher macro EI value of 1.91 compared to the 1.84 reported in Hoel et al.'s study. This discrepancy arises from the fact that the macro TPM in \cite{Hoel2013} encompasses 9 states, whereas our findings involve only 7 states. In their approach, \{000111\} and \{111000\} are treated as distinct additional macro states \{02\} and \{20\}, whereas we consider them as a single macro-state \{11\}.}
   \label{fig:SI_ex2}
\end{figure}

\begin{figure}[ht!]
    \centering
    \includegraphics[width=1\linewidth]{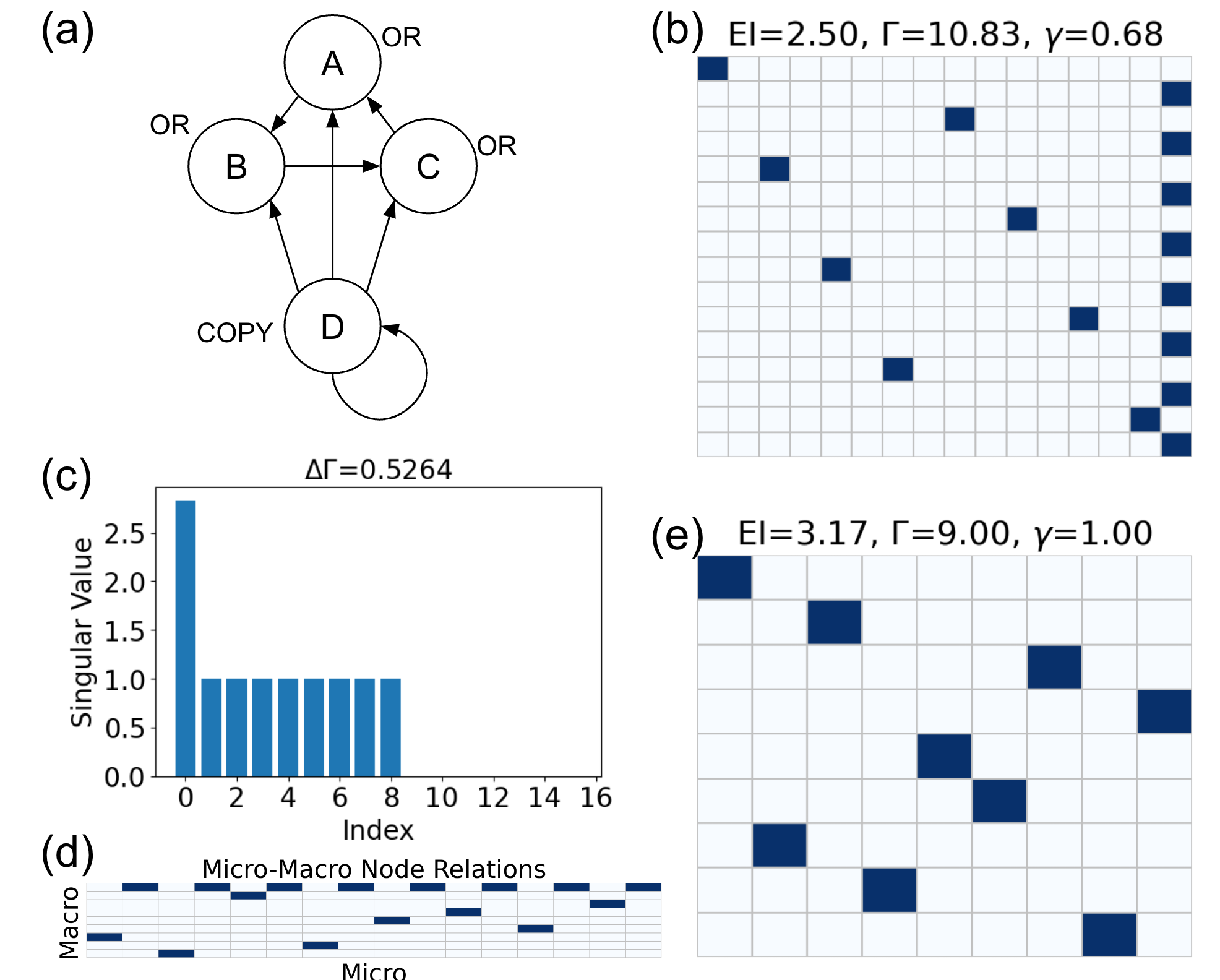}
    \caption{(a) is a directed Boolean network. (b) is the TPM of (a). (c) is the singular spectrum of (b). (d) is projection matrix. (e) is the course-grained TPM. We also get a larger macro EI than Hoel's (3.17 here and 3.00 in \cite{Hoel2017}). It is because we have 9 macro states and \cite{Hoel2017} only has 8. We combined all the 8 micro states that transform to \{11111111\} as a separate macro state, while \cite{Hoel2017} combines them into each of the other 8 groups.}
    \label{fig:SI_ex3}
\end{figure}

\newpage

% \color{red}
\section{\label{sec:commutativity}The Proof of the Commutativity of our Coarse-Graining Strategy}

In section \ref{sec:coarse-graining}, we give a coarse-graining method based on the SVD and the stationary distribution. One of the most advantages of this method is that the commutativity of the coarse-graining operator and the time evolution operator(the TPM). In this section, we formalize this characteristic by the following proposition.

\begin{proposition}
\label{thm:commutativity}
For a given Markov chain $\chi$ with the corresponding TPM $P$ and its stationary distribution $\mu$, and for a given clustering projection matrix $\Phi$ defined in \ref{eq:projection_matrix}, we coarse-grain the TPM $P$ according to method mentioned in Section \ref{sec:coarse-graining} to derive a macro-level TPM represented by $P'$. Then, the commutativity between the operators of dynamics ($P$ or $P'$) and the coarse-graining ($\Phi$) holds, i.e., the following equation is satisfied:
\begin{equation}
    P\cdot \Phi=\Phi\cdot P'.
\end{equation}
\end{proposition}
\begin{proof}

Because $\mu$ is the stationary distribution of $P$, thus,
\begin{equation}
    \label{eqn:stationary_definition}
    \mu=\mu \cdot P,
\end{equation}

And we can define a stationary flow matrix $F$ (Equation \ref{eq:stationary_flow}):
\begin{equation}
\begin{aligned}
    \label{eq:F113}
    F \equiv diag{(\mu)} \cdot P.
\end{aligned}
\end{equation}

According to Section \ref{sec:coarse-graining}, the coarse-grained stationary flow matrix $F'$ is computed as (Equation \ref{eq:F_definition}):
\begin{equation}
    \label{eq:F114}
    F'=\Phi^T\cdot F\cdot\Phi.
\end{equation}

The corresponding coarse-grained stationary distribution $\mu'$ can be defined as:
\begin{equation}
    \label{eq:F115}
    diag{(\mu')}=\Phi^T \cdot diag{(\mu)} \cdot \Phi.
\end{equation}

And the reduced TPM $P'$ according to Equation \ref{eq:coarse-graining} can be written as:
\begin{equation}
\begin{aligned}
    \label{eq:F116}
    P'_i = F'_i/\sum_{j=1}^{r}(F'_i)_j = F'_i/\mu'_i, \forall i\in\{1,2,\cdots,r\},
\end{aligned}
\end{equation}
where, the second equality holds because
\begin{equation}
\begin{aligned}
    \sum_{j=1}^{r}(F'_i)_j
    &= (F' \cdot \mathbbm{1}_{r \times 1})_i\\
    &\overset{(\ref{eq:F114})}{=} (\Phi^T \cdot F \cdot \Phi \cdot \mathbbm{1}_{r \times 1})_i = (\Phi^T \cdot F \cdot \mathbbm{1}_{n \times 1})_i \\
    &\overset{(\ref{eq:F113})}{=} (\Phi^T \cdot diag{(\mu)} \cdot P \cdot \mathbbm{1}_{n \times 1})_i = (\Phi^T \cdot diag{(\mu)} \cdot \mathbbm{1}_{n \times 1})_i = \mu'_i.
\end{aligned}
\end{equation}
In which, the third equality holds because there is one 1 in each row in $\Phi$ and other elements are 0, and the firth equality holds because the sum of each row in $P$ equals 1.

And the matrix form of Equation \ref{eq:F116} can be written as:
\begin{equation}
 P' = (diag{(\mu')})^{-1} \cdot F'.
\end{equation}
Therefore,
\begin{equation}
\begin{aligned}
\label{eq:117}
P \cdot \Phi 
& \overset{(\ref{eq:F113})}{=} (diag{(\mu)})^{-1} \cdot F \cdot \Phi\\
& \overset{(\ref{eq:F114})}{=} (diag{(\mu)})^{-1} \cdot (\Phi^T)^{-1} F'\\
& \overset{(\ref{eq:F116})}{=} (diag(\mu))^{-1} \cdot (\Phi^T)^{-1} diag{(\mu')} \cdot P'\\
& \overset{(\ref{eq:F115})}{=} (diag(\mu))^{-1} \cdot (\Phi^T)^{-1} \Phi^T \cdot diag{(\mu)} \cdot \Phi \cdot P'\\
& = \Phi \cdot P'
\end{aligned}
\end{equation}

\end{proof}

This shows that the coarse-graining method we have used satisfies:
\begin{equation}
\label{eqn:last}
P' = (diag{(\mu')})^{-1} \cdot F' = (\Phi^T \cdot diag{(\mu)} \cdot \Phi)^{-1} \cdot \Phi^T\cdot diag{(\mu)} \cdot P \cdot\Phi,
\end{equation}
and it ensures the commutativity for arbitrary clustering projection $\Phi$.

\end{appendices}

% \bibliographystyle{plainnat}


\end{document}